\documentclass[10pt]{IEEEtran}

\usepackage{cite}
\usepackage[cmex10]{amsmath}
\usepackage{amssymb}
\usepackage{graphicx}
\usepackage[caption=false,font=footnotesize]{subfig}
\usepackage{algorithm,algorithmic}

\newtheorem{theorem}{Theorem}
\newtheorem{lemma}{Lemma}
\newtheorem{remark}{Remark}

\newcommand\relphantom[1]{\mathrel{\phantom{#1}}}

\begin{document}

\title{Distributed Parameter Estimation with Quantized Communication via Running Average\thanks{This research is funded by the Republic of Singapore's National Research Foundation through a grant to the Berkeley Education Alliance for Research in Singapore (BEARS) for the Singapore-Berkeley Building Efficiency and Sustainability in the Tropics (SinBerBEST) Program. BEARS has been established by the University of California, Berkeley as a center for intellectual excellence in research and education in Singapore. A preliminary version of this
paper was presented at ICASSP 2015, Brisbane, Australia \cite{ZhuSohXieLiu15}.}}

\author{Shanying~Zhu, Yeng Chai Soh,~\IEEEmembership{Senior~Member,~IEEE,} and~Lihua Xie,~\IEEEmembership{Fellow,~IEEE} \\
\thanks{Copyright (c) 2015 IEEE. Personal use of this material is permitted. However, permission to use this material for any other purposes must be obtained from the IEEE by sending a request to pubs-permissions@ieee.org.}
\thanks{The authors are with Centre for System Intelligence and Efficiency (EXQUISITUS), School of Electrical and Electronic Engineering, Nanyang Technological University, 50 Nanyang Avenue, Singapore 639798
(E-mail: syzhu@ntu.edu.sg,  eycsoh@ntu.edu.sg, elhxie@ntu.edu.sg).}
}

\markboth{IEEE Transactions on Signal Processing, DOI: 10.1109/TSP.2015.2441034}%
{S. Zhu \MakeLowercase{\textit{et al.}}: Distributed Parameter Estimation with Quantized Communication}

\maketitle

\begin{abstract}
	In this paper, we consider the problem of parameter estimation over sensor networks in the presence of quantized data and directed communication links. We propose a two-stage distributed algorithm aiming at achieving the centralized sample mean estimate in a distributed manner. Different from the existing algorithms, a running average technique is utilized in the proposed algorithm to smear out the randomness caused by the probabilistic quantization scheme. With the running average technique, it is shown that the centralized sample mean estimate can be achieved both in the mean square and almost sure senses, which is not observed in the standard consensus algorithms. In addition, the rates of convergence are given to quantify the mean square and almost sure performances. Finally, simulation results are presented to illustrate the effectiveness of the proposed algorithm and highlight the improvements by using running average technique.  
\end{abstract}

\IEEEpeerreviewmaketitle

\IEEEpeerreviewmaketitle

\section{Introduction}
	Sensor networks, composed
	of a large number of signal processing devices (nodes), are massively distributed systems for sensing and processing of spatially dense data with wide applications both in military and civilian scenarios.
	A  popular application of sensor networks is the decentralized estimation of
	unknown parameters using samples collected from nodes \cite{KamFarRoy13,SchRibGia08,XiaRibLuo_etal06,XiaoBoyLal05}. Two prevailing topologies for such task are fusion center based networks and ad hoc networks \cite{XiaRibLuo_etal06}. Compared with fusion  center based networks, ad hoc networks have several advantages including  scalability and resilience of node failure. In
	a typical estimation problem in ad hoc networks, nodes make noisy measurements
	of variables of interest. The main concern is how to utilize
	the samples to produce a desired estimate by only
	exchanging data between neighboring nodes.

	Distributed estimation in ad hoc networks is usually based on successive
	refinements of local estimates maintained at individual nodes. In most applications, nodes are powered by batteries with finite lifetime and thus have limited computing and communication capabilities. Another aspect is bandwidth constraint, which renders the transmission of large volume of real-valued data impractical. This means that the data exchanged between nodes need to be quantized prior to transmission. However, this process introduces certain quantization errors which could have severe effects. The errors will be accumulated throughout the successive iterations, making the estimation process fluctuating or even divergent \cite{XiaoStepKim07}. 
    
    A number of distributed consensus algorithms have been proposed to address the problem of estimation with quantized communication. Most of them assume symmetric communication between nodes. Actually, in ad hoc networks, communication links between certain pairs of nodes may be directed, i.e., a node can receive information from another node but not vice versa. This could be caused by non-homogeneous interference, packet collision and so on.   Motivated by this observation, in the paper, we consider the problem of distributed estimation over directed topologies and examine its convergence behavior under the effect of quantized communication.

\subsection{Related work}
	Distributed consensus algorithms are effective ways to solve the estimation problems in sensor networks, where the final states are mostly chosen as the estimates. Recently, much attention has been paid to the effect of quantization on consensus algorithms. For instance, deterministic quantization schemes are used in \cite{CarFagFraZam10,ChaLiuBas14,LiuLiXieetal13,NedOlsOzdTsi09}. In particular, uniform and truncation quantizers were investigated in \cite{CarFagFraZam10,ChaLiuBas14,NedOlsOzdTsi09}, where convergence can only be guaranteed up to a neighborhood of the target average and upper bounds characterizing the gaps were provided. 
	Ref.~\cite{LiuLiXieetal13} considered the logarithmic quantization scheme, which showed that the consensus error is upper bounded by a quantity depending on the quantization resolution and initial states. In \cite{KasBasSri07}, a quantized consensus algorithm was introduced with an additional constraint that the states of the nodes are integers. This constraint leads to an integer approximation of the target average. Extension to the directed topologies has been examined in \cite{CaiIsh11}. 

	Another thread is to adopt probabilistic quantization schemes. In \cite{AysCoatRabb08}, the dithered quantization scheme was used. It was shown that consensus to a random variable whose expectation is equal to the desired average can be reached almost surely. This kind of convergence was also observed for gossip algorithms \cite{CarFagFraZam10}. In fact, even employing the decaying link weights satisfying a persistence condition cannot guarantee the convergence to the target average \cite{KarMoura10}. The quantization scheme introduced in \cite{FangLi09} adaptively adjusts the quantization threshold and step-size by learning from previous runs, in a way such that the target average can be achieved in the mean square sense.
	Another method that achieves the target average is to explore the temporal information of the successive state \cite{FangLi10}. Most of the above works assume that the communication topology
	is symmetric, which may not be realistic as discussed previously. Moreover, the symmetric requirement imposes much effort on the nodes to acquire necessary topology information to construct the weight matrices. Even if the symmetric communication is assumed, the aforementioned results indicate that convergence to the target average  is not possible in most cases using simple quantizers.

	To further address the residual issue of quantization, dynamic encoding/decoding schemes were proposed  in \cite{CarlBullZamp10,LiFuXieZhang11} to ensure the convergence to the desired average value. Specifically,  Ref. \cite{LiFuXieZhang11} showed that the number of quantization bits can be reduced  to merely one by appropriately designing the scaling function and some control parameters. The result of \cite{LiFuXieZhang11} has been extended to directed graphs in \cite{LiLiuWangLin13}, where the weighted average instead of the desired average was shown to be  achievable. Although dynamic quantizations perform quite well, some spectral properties of the Laplacian matrix of the underlying topology have to be known in advance based on which the encoder-decoder parameters are carefully chosen.  A similar idea was adopted in \cite{ThanKokPuFros13} to design a progressive quantizer that progressively reduces the quantization intervals during the convergence of the algorithm.

  \subsection{Summary of contributions} 
  	In this paper, we consider the problem of parameter estimation over directed communication topologies. Each node  has real-valued states but can only exchange information with its neighbors utilizing quantized communication. 
  	The main contributions are summarized as follows:

  	Firstly, we propose a two-stage distributed estimation algorithm in which the nodes utilize basic probabilistic quantization.  At the first stage, we estimate the left eigenvector with respect to the zero eigenvalue of the Laplacian matrix. This information is then used at the second stage to construct a correction term aiming at compensating for the unidirectional effect of directed communication links.  At both stages, the running average technique is utilized to limit the quantization effect on the estimation process.  Unlike \cite{AysCoatRabb08,CarFagFraZam10,ChaLiuBas14,KarMoura10,LiuLiXieetal13,NedOlsOzdTsi09}, our algorithm does not require the weight matrix to be doubly stochastic. And it can be run over any strongly connected topology without any knowledge of the out-neighbor information and the left eigenvector of the corresponding Laplacian matrix  as required by those in \cite{CaiIsh11,LiLiuWangLin13}.

	Secondly, a comprehensive convergence analysis of the proposed algorithm is given. With the running average technique, we show that the centralized sample mean estimate can be achieved exactly both in the mean square and almost sure senses. The results extend the one in \cite{FangLi10} from undirected graphs to directed graphs. Moreover, the proposed algorithm does not depend on the complicated design of quantization schemes as in \cite{CarlBullZamp10,FangLi09,LiFuXieZhang11,LiLiuWangLin13,ThanKokPuFros13}. Our analysis relies on the theoretical tools of the laws of large numbers and the iterated logarithm. The theoretical results reveal that simple quantization schemes can be employed to solve the parameter estimation problems over networks, provided that a suitable form of estimator is introduced.  

    The paper is organized as follows: In Section~\ref{sec:problem}, we present the problem formulation and some preliminary results needed in the subsequent sections. In Section~\ref{sec:distributed_estimation_algorithm_on_digraphs}, we describe the proposed two-stage distributed algorithm along with some implementation considerations. Convergence analyses both in the mean square and almost sure senses are presented in  Section~\ref{sec:convergence_analysis}. 
    Section~\ref{sec:simulation_results} presents the simulation results to illustrate the effectiveness of the proposed algorithm, followed by the conclusions and future works in Section~\ref{sec:concluding_remarks}.

	\emph{Notation:} $\mathbb{Z}_{\geq a}$ stands for the subset of integers greater than $a$. For two functions $f(k)$ and $g(k)$, $f(k)=o(g(k))$ means that $\lim_{k\to \infty} f(k)/g(k)=0$. 	
	We will drop $o(g(k))$ in $g(k)+o(g(k))$ if no ambiguity arises.  We use $O(1)$ to denote a constant, which may vary at different places.
	$\mathbb{R}^{m\times
	n}$ denotes the set of all $m\times n$ matrices with the Euclidean norm $\|\cdot\|_2$ and Frobenius norm $\|\cdot\|_F$ with compatible vector norm $\|\!\cdot\!\|$. We use bold uppercase and lowercase letters to denote matrices and vectors, respectively. $\mathbf{I}$ is the identity matrix, $\mathbf{1},\mathbf{0}$ are all-one and all-zero vectors, respectively.
	$\lambda_{\max}(\cdot)$ represents the largest eigenvalue of a symmetric matrix.  For a random  vector $\mathbf{x}$, $\mathbb E\{\mathbf{x}\}$ denotes its expectation and $\text{Cov}(\mathbf{x})$ its covariance. 

\section{Problem formulation}\label{sec:problem}

	Consider the estimation problem in a sensor network consisting of $n$ homogeneous nodes, each making observations of an unknown parameter $\theta\in \mathbb{R}$. The observations are corrupted by additive noises, i.e.,
\[
	y_i=\theta+w_i, \ i=1, 2, \dots,n,
\]
	where $\{w_i\}_{i=1}^n$ are zero mean, i.i.d. Gaussian noises. If there is a fusion center having access to all the samples $\{y_i\}_{i=1}^N$, then the  sample mean estimator $\hat{\theta} \triangleq (1/n)\sum_{i=1}^n y_i$ is the best one in the sense of Cram\'er-Rao lower bound \cite[p.30]{Kay93}. This estimator is universal since it does not require any information of the noise \cite{Luo05}.

	The distributed estimation problem is concerned with computing the centralized sample mean estimate 
	$\hat{\theta}$ iteratively at every node without requiring global knowledge of $\{y_i\}_{i=1}^n$ and the network topology.
	We model the communication topology over which the nodes exchange information as a \emph{weighted directed graph} $\mathcal{G}=(\mathcal{V},\mathcal{E},\mathbf{A})$, where 
	$\mathcal{V}=\{1,2,\dots,n\}$ is the set of nodes, $\mathcal{E}\subset\mathcal{V}\times\mathcal{V}$ denotes all the unidirectional communication links between nodes and $\mathbf{A}=[a_{ij}]_{n\times n}$ is composed of weights $a_{ij}>0$ associated with each directed edge $(j,i)\in \mathcal{E}$. It is assumed that there are no self-loops in $\mathcal{G}$. The directed edge $(j,i)$ means that node $i$ can receive data from node $j$. We denote $\mathcal{N}_i=\{j: (j,i)\in \mathcal{E}\}$ as the set of \emph{neighbors} of node $i$. We make the following assumption:

	\emph{Assumption 1:} Graph $\mathcal{G}$ is strongly connected, i.e., for any two nodes $i$ and $j$, there exists a directed path from $i$ to $j$.

	In the case of limited communication rate between nodes, each node will first quantize the data prior to its transmission to the neighbors. In this paper, we adopt the following estimation algorithm at each node $i$,
\begin{align}\label{eq:algorithm}
	x_i(t+1)&=\hat{x}_i(t)+\alpha\sum_{j\in\mathcal{N}_i} a_{ij}[\mathcal{Q}(\hat{x}_j(t))-\mathcal{Q}(\hat{x}_i(t))],
\end{align}
	with initial guess $x_i(t_0)=y_i$, where $\alpha>0$ is a constant,  $\mathcal{Q}(\cdot)$ denotes the quantization operation and 
\begin{equation}
	\hat{x}_i(t)\triangleq x_i(t)+\epsilon_i(t),
\end{equation}
	in which $\epsilon_i(t)$ is a correction term to compensate for the unidirectional effects of communication links. The goal is to design an appropriate $\epsilon_i(t)$   such that all the nodes can asymptotically acquire  the centralized $\hat{\theta}$  over any strongly connected topology.

\begin{remark}\label{rem:correction}
	For standard consensus algorithms without quantization, i.e., $x_i(t+1)=x_i(t)+\alpha\sum_{j\in\mathcal{N}_i} a_{ij}[x_j(t)-x_i(t)]$, it is well known that the state $x_i(t)$ will converge to the weighted average of $y_i$ rather than $\hat{\theta}$, where the weights are determined by the spectral knowledge of  graph $\mathcal{G}$. The introduction of the correction term $\epsilon_i(t)$  in \eqref{eq:algorithm} is meant to drive the weighted average to the sample mean estimate.
\end{remark}

\begin{remark}\label{rem:algorithm}
	The algorithm \eqref{eq:algorithm} belongs to the compensating update rule \cite{FrasCarFagZamp09,CarFagFraZam10,FangLi10,ThanKokPuFros13}, where both the real-valued states and their quantized values are used to compute the states at next step. This strategy is meant to fully exploit the implicit channel feedback which comes from quantization. 
\end{remark}

\subsection{Probabilistic quantization}
	We present a brief review of the quantization scheme used in the paper. Each node is equipped with a probabilistic quantizer $\mathcal{Q}(\cdot): \mathbb{R}\to \mathcal{S}_{\Delta}$ with the set of quantization levels $\mathcal{S}_{\Delta}= \{k\Delta: k\in \mathbb{Z}\}$, where $\Delta$ is the quantization step-size. For any $x\in \mathbb{R}$, it is  quantized in a probabilistic manner: 
\[
	 \mathcal{Q}(x)=\begin{cases} 
	 \left\lceil \frac{x}{\Delta}\right\rceil \Delta,&\text{with probability} \ p,\\
	 \left\lfloor \frac{x}{\Delta}\right\rfloor \Delta, & \text{with probability}\ 1-p,
	 \end{cases}
\]
	where $p=x/\Delta-\lfloor x/\Delta\rfloor$, $\lfloor \cdot \rfloor$ and $\lceil \cdot \rceil$ denote the floor and ceiling functions, respectively. 
	We can prove  that the quantized message $\mathcal{Q}(x)$ is an unbiased estimator of $x$ with finite variance \cite{AysCoatRabb08,CarFagFraZam10}, that is,
\begin{equation}\label{eq:varianncequantierror}
	\mathbb{E}\{\mathcal{Q}(x)\}=x, \ \ \mathbb{E}\left\{(\mathcal{Q}(x)-x)^2\right\}\leq \frac{\Delta^2}{4}.
\end{equation}
	Further, it is obvious that
\begin{equation}\label{eq:boundquantierror}
	|\mathcal{Q}(x)-x|\leq \Delta. 
\end{equation}

	Actually, the above quantization is equivalent to a substractively dithered method \cite{AysCoatRabb08}. If the dither sequence satisfies the Schuchman conditions, then the quantization errors are statistically independent from each other and the input~\cite{LipWannVan92}. 	We make the following natural assumption of statistical independence:

	\emph{Assumption 2}: The quantization errors are independent from the data, and are  temporally\footnote{The spatial independence of quantization errors is introduced to ease the notation. All the results can be easily extended to the non-spatial case.} and spatially independent. 

\subsection{Averaging technique} 
	Existing results in \cite{AysCoatRabb08,CarFagFraZam10,KarMoura10} reveal that the state of consensus algorithms is not a qualified estimator in the case of basic probabilistic quantization, as there is always residue between the final state and $\hat{\theta}$ unless certain adaptive mechanism is adopted \cite{FangLi09,ThanKokPuFros13}. We need to find an appropriate  form of estimator to tackle the quantization issue.

	Statistics tells us that large samples have smoothing effects: The wild randomness that always exists in small samples will be smeared out \cite[p.201]{Gut13}. By Assumption~2, the quantization errors are temporally independent. This temporal information has been used in \cite{FangLi10} to investigate the consensus seeking over undirected graphs, which motivates us to adopt the following running average to smooth the samples
\begin{equation}\label{eq:arithmeticmean}
	\bar{x}_i(K)\triangleq \frac{1}{K} \sum_{k=t_0+1}^{t_0+K}x_i(k), \ \ \forall i=1,2,\dots,n.
\end{equation}
    The new quantity $\bar{x}_i(K)$ will be used as the estimate of the unknown parameter $\theta$ at node $i$. This  formulation of the distributed estimation problem over sensor networks differs from the standard consensus algorithms, where the focus is the performance of  $x_i(k)$ for consensus algorithms.

\subsection{Preliminaries}
	One important concept for distributed algorithms is the Laplacian $\mathbf{L}$ corresponding to graph $\mathcal{G}$, which is defined as $\mathbf{L}\triangleq \mathbf{D}- \mathbf{A}$, where $\mathbf{D}\triangleq \text{diag}\{d_1,d_2,\dots,d_n\}$ and $d_i=\sum_{j\in \mathcal{N}_i} a_{ij}$, $\forall i$. It is clear that $\mathbf{L} \mathbf{1}=\mathbf{0}$, that is, 0 is an eigenvalue of $\mathbf{L}$.

\begin{lemma}\label{lem:matrixproperty}
	Let $\boldsymbol{\omega}=[\omega_1,\omega_2,\dots,\omega_n]^T$ be the left eigenvector corresponding to the zero eigenvalue of $\mathbf{L}$ with $\mathbf{1}^T \boldsymbol{\omega}=1$. Then under Assumption~1, $\boldsymbol{\omega}$ is positive and the matrix $\mathbf{Q}\triangleq \mathbf{P}-\mathbf{1}\boldsymbol{\omega}^T$ with $\mathbf{P}\triangleq \mathbf{I}-\alpha \mathbf{L}$ has the following properties:
\begin{enumerate}
	\item[i)] \emph{Spectrum}: 
	Let $0=\lambda_1(\mathbf{L}), \lambda_2(\mathbf{L}),\dots, \lambda_n(\mathbf{L})$ be the eigenvalues of Laplacian $\mathbf{L}$, then the spectrum of $\mathbf{Q}$ is $\{0,1-\alpha \lambda_i(\mathbf{L}),i=2,3,\dots,n\}$;
	\item[ii)] \emph{Spectral radius}: 
	The spectral radius  $\rho(\mathbf{Q})<1$ if and only if  $0<\alpha<\min_{2\leq i\leq n}\left\{2 \text{Re}(\lambda_i(\mathbf{L}))/|\lambda_i(\mathbf{L})|^2\right\}$, where $\text{Re}(\lambda_i(\mathbf{L}))$ represents the real part of $\lambda_i(\mathbf{L})$;
	\item[iii)] \emph{Bounds on Frobenius norm}: 
	The Frobenius norm of power $\mathbf{Q}^k$, $\forall k\in \mathbb{Z}_{\geq 1}$,  is bounded by
\[
	\|\mathbf{Q}^k\|_F\leq n c_\mathbf{Q} k^{q-1}\rho^k(\mathbf{Q}),
\]
	where $c_\mathbf{Q}>0$ is a constant depending only on $\mathbf{Q}$ and $q\triangleq \max_{\lambda_i(\mathbf{Q})\neq 0}\{q_i\}$, $q_i$ is the multiplicity of $\lambda_i(\mathbf{Q})$ in the
	minimal polynomial of $\mathbf{Q}$.	
\end{enumerate}
\end{lemma}
\begin{proof}
	See Appendix~\ref{app:matrixproperty}.
\end{proof}

	The next lemma presents a way to choose the parameter $\alpha$ such that $\mathbf{Q}$ has some desired properties as given in Lemma~\ref{lem:matrixproperty}.
\begin{lemma}\label{lem:parameter}
	Let $0<\alpha<1/\max_{i} d_i$, then under Assumption~1,  we have $\rho(\mathbf{Q})<1$. Further, for all $k\in \mathbb{Z}_{\geq 1}$,
\[
	\bigl\|\mathbf{I}-\mathbf{Q}^k\bigr\|_F\leq \sqrt{n+2+n^2 c_\mathbf{Q}^2 k^{2(q-1)}\rho^{2k}(\mathbf{Q})}\leq c_{\mathbf{Q},n},
\]
   where $c_{\mathbf{Q},n}^2\triangleq n+2+n^2c_{\mathbf{Q}}^2((1-q)/(e\log \rho(\mathbf{Q})))^{2(q-1)}$.
\end{lemma}
\begin{proof}
	See Appendix~\ref{app:parameter}.
\end{proof}

\section{Distributed estimation algorithm over directed topologies via Running Average} 
\label{sec:distributed_estimation_algorithm_on_digraphs}
	In this section, the averaging technique proposed in the previous section is applied to the estimation problem to achieve the centralized sample mean estimate in a distributed manner over directed communication topologies. 

    Different from undirected communication topologies, the primary challenge of achieving centralized sample mean estimate over directed sensor networks lies in that the state sum of nodes needs not be preserved, thereby causing shifts in the average. In fact, in this case, only a weighted version of the sample mean estimate, i.e., $\boldsymbol{\omega}^T \mathbf{y}\neq \hat{\theta}$, can be asymptotically obtained \cite{RenBea05,LiLiuWangLin13}, where $\boldsymbol{\omega}$ is the left eigenvector of $\mathbf{L}$  associated with the zero eigenvalue and $\mathbf{y}=[y_1,\dots,y_N]^T$. We note that several techniques have been proposed in the literature to tackle the issue of directed  topologies for consensus algorithms. In these algorithms, either an extra variable is associated with each node by assuming some out-neighbor information \cite{BenBlonThiTsiVet10,CaiIsh12,DomHad11} or certain compensation mechanism related with the left eigenvector $\boldsymbol{\omega}$ is performed~\cite{PriGasetal14}.

	In this paper, we follow the latter approach and borrow some ideas from \cite{PriGasetal14} to deal with the unidirectional effect arising from directed topologies. The main advantage of the method is that we do not need any knowledge of the out-neighbor information as required by those in \cite{BenBlonThiTsiVet10,CaiIsh12,DomHad11}. The proposed algorithm is composed of two stages: At the first stage, we apply the averaging technique to estimate the left eigenvector $\boldsymbol{\omega}$;  At the second stage, we design the correction term $\epsilon_i(t)$ in \eqref{eq:algorithm} to compensate for the effect of the directed links by using estimates  obtained at the first stage. A distributed estimation algorithm via interwinding these two stages is then proposed.

\subsection{Distributed estimation of the left eigenvector $\boldsymbol{\omega}$}
	At the first stage, each node $i$ maintains a vector $\mathbf{z}_i=[z_{i1},z_{i2},\dots,z_{in}]^T$ to store the estimate of $\boldsymbol{\omega}$.  At each iteration, the nodes update their variables as follows:
\begin{equation}\label{eq:estimateeigenvec}
	\mathbf{z}_i(t+1)=\mathbf{z}_i(t)+\alpha\sum_{j\in\mathcal{N}_i} a_{ij} [\mathcal{Q}(\mathbf{z}_j(t))-\mathcal{Q}(\mathbf{z}_i(t))],
\end{equation}
	with initial values $z_{ii}(0)=1$, $z_{ij}(0)=0$, $\forall j\neq i$, where $0<\alpha<1/\max_i d_i$ and $\mathcal{Q}(\cdot)$ is componentwise for vectors. 

	In order to ensure that all nodes can achieve reliable estimates of $\boldsymbol{\omega}$, it suffices to guarantee that $\mathbf{Z}(t)=[\mathbf{z}_1(t),\mathbf{z}_2(t),\dots, \mathbf{z}_n(t)]^T$ converges to $\mathbf{1}\boldsymbol{\omega}^T$. This is true if there are no quantization errors \cite[Theorem 8.4.4]{HorJoh13}. However, it is no longer the case if some quantization errors are present.  Actually, using the similar arguments as in \cite{FrasCarFagZamp09}, we note that  the protocol \eqref{eq:estimateeigenvec} can only  converge to a neighborhood of $\mathbf{1}\boldsymbol{\omega}^T$ with non-vanishing errors.
	Fortunately, the temporal information of the quantization errors as assumed in Assumption~2 can be exploited, of which the temporal independence enables the running average to fall under the purview of the law of large numbers of independent random vectors. This motivates us to apply the averaging technique discussed in Section~\ref{sec:problem} to  remove noise propagation, and adopt 
\[
   \bar{\mathbf{Z}}(K)=[\bar{z}_{ij}(K)]_{n\times n}\triangleq \frac{1}{K}\sum_{k=k_0+1}^{k_0+K} \mathbf{Z}(k)
\]
    as the estimate of $\mathbf{1}\boldsymbol{\omega}^T$.

    The above discussion leads to the proposed Algorithm~\ref{alg:estlefteig}.
    In the algorithm, we use the initial value $z_{ii}(0)=n^{\kappa}$ with $\kappa\geq 0$ instead of the original $z_{ii}(0)=1$. One reason is that convergence of the original $\bar{z}_{ii}$ to $\omega_i$ is equivalent to its convergence to $n^{\kappa}\omega_i$ in the new scale. Introducing $n^{\kappa}$ into the initial values does not affect the convergence.

\renewcommand{\algorithmicrequire}{ \textbf{Input:}} %Use Input in the format of Algorithm  
\renewcommand{\algorithmicensure}{ \textbf{Output:}} %UseOutput in the format of Algorithm 
\begin{algorithm}[!h]
	\caption{Distributed estimation of  $\boldsymbol{\omega}$ at node $i$}\label{alg:estlefteig}
\begin{algorithmic}[1]
	\REQUIRE $\alpha$,  $n$, $\kappa$, $a_{ij}$, $k_0$.
	\ENSURE $\bar{\mathbf{z}}_i/n^{\kappa}$.
	\STATE \textbf{Initialization:} \label{alg:initialz}
	$z_{ii}(0)=n^{\kappa}$, $z_{ij}(0)=0$, $\forall j\neq i$.
	\STATE Receive data from neighbors: $\mathcal{Q}(\mathbf{z}_j(t))$, $j\in \mathcal{N}_i$.
	\STATE Update the estimate of $\boldsymbol{\omega}$ via \eqref{eq:estimateeigenvec}.
	\IF{$t\geq  k_0$}
	\STATE $K\triangleq t-k_0$.
	\STATE Update the average $\bar{\mathbf{z}}_i(K)$:\\[1mm] \label{alg:updatez_i}
	$\bar{\mathbf{z}}_i(K+1)=\frac{K}{K+1} \bar{\mathbf{z}}_i(K)+\frac{1}{K+1} \mathbf{z}_i(t)$. 
	\ENDIF
\end{algorithmic}
\end{algorithm}

\subsection{Design of the correction term $\boldsymbol{\epsilon}(t)$} 
\label{sub:algorithm_description}
	The second stage is concerned with the design of an appropriate correction term $\epsilon_i(t)$ in \eqref{eq:algorithm} for each $i$ to compensate for the unidirectional effect of directed communication links.

	As discussed previously, the steady state of the algorithm \eqref{eq:algorithm} is closely related with the left eigenvector $\boldsymbol{\omega}$ of the Laplacian $\mathbf{L}$. Now assume that $\boldsymbol{\omega}$ is available at each node, then the nodes can adjust their initial values as  $x_i(t_0)=y_i+\zeta_i, \ \forall i\in\{1,\dots,n\}$
	with $\zeta_i=y_i[1/(n\omega_i)-1]$ so that $\boldsymbol{\omega}^T \mathbf{x}(t_0)=\hat{\theta}$, and thus $\hat{\theta}$ can be asymptotically achieved at all nodes. However, Algorithm~\ref{alg:estlefteig} can only produce an asymptotic estimate of $\boldsymbol{\omega}$ (see Theorems~\ref{thm:MSconvergence} and \ref{thm:eZas}).
	A possible alternative is to perform the tuning via $\epsilon_i(t)$ in an iterative manner upon the estimate $\bar{\mathbf{z}}_i(K)$ of $\boldsymbol{\omega}$ is obtained so that $\boldsymbol{\omega}^T \mathbf{x}(t)\to \hat{\theta}$ as $t\to \infty$. 

    To design an appropriate form, we note that $\boldsymbol{\omega}^T \mathbf{L}=\mathbf{0}$, which implies $\boldsymbol{\omega}^T \mathbf{x}(t+1)=\boldsymbol{\omega}^T (\mathbf{x}(t)+\boldsymbol{\epsilon}(t))$ in view of \eqref{eq:algorithm}, where $\mathbf{x}(t)$ and $\boldsymbol{\epsilon}(t)$ are the stacked vectors of $x_i(t)$ and $\epsilon_i(t)$, respectively. This implies that 
\[
	\boldsymbol{\omega}^T \mathbf{x}(t)=\boldsymbol{\omega}^T \left(\mathbf{y}+\sum_{s=t_0}^{t-1} \boldsymbol{\epsilon}(s)\right).
\]
    In order to guarantee the convergence of $\sum_{s=t_0}^{t-1} \boldsymbol{\epsilon}(s)$, one simple choice of $\boldsymbol{\epsilon}(t)$ is to make $\sum_{s=t_0}^{t-1} \boldsymbol{\epsilon}(s)$ a telescoping series. For instance, we can design $\epsilon_i(t)$ as  follows
\begin{equation}\label{eq:correction}
	\epsilon_i(t)\triangleq \begin{cases}
    \bigl[\frac{1}{n\bar{z}_{ii}(t_0+1)}-1\bigr]y_i, &t=t_0,\\
	\bigl[\frac{1}{n\bar{z}_{ii}(t+1)}- \frac{1}{n\bar{z}_{ii}(t)}\bigr]y_i, &t\in \mathbb{Z}_{\geq t_0+1}.
	\end{cases}
\end{equation}
    In this way, substituting \eqref{eq:correction} into \eqref{eq:algorithm} yields
\begin{equation}
	\boldsymbol{\omega}^T \mathbf{x}(t)=\frac{1}{n}\left[\frac{\omega_1}{\bar{z}_{11}(t)},\dots,\frac{\omega_n}{\bar{z}_{nn}(t)}\right]\mathbf{y},\label{eq:compensate}
\end{equation}
    which will asymptotically converge to $\hat{\theta}$ as $t\to \infty$ provided that the convergence of Algorithm~\ref{alg:estlefteig} is established. 

	One issue remaining before the implementation of \eqref{eq:correction} is the well-definedness of $\epsilon_i(t)$, $\forall i$, that is, the denominators in $\epsilon_i(t)$ must be nonzero with probability 1. This is much involved and we will elaborate on it in Section~\ref{sec:convergence_analysis} (see Theorem~\ref{thm:welldefined}). 

    The proposed algorithm of the $t$-th iteration run by node $i$ at the second stage is shown in Algorithm~\ref{alg:estimation}. Here, we modify the definition of $\epsilon_i(t)$ in \eqref{eq:correction}  to accommodate the setup in Algorithm~\ref{alg:estlefteig} (see lines \ref{alg:correctionintial} and \ref{alg:correction}). 
    Moreover, a running average step as in Algorithm~\ref{alg:estlefteig} is introduced aiming at removing noise propagation (see line \ref{alg:updatex_i}).

\begin{algorithm}[!h]
	\caption{Distributed estimation algorithm with quantized data via running average at node $i$}\label{alg:estimation}
\begin{algorithmic}[1]
	\REQUIRE $\alpha$,  $n$, $\kappa$, $a_{ij}$, $t_0$, $x_i(t_0)$, $\bar{z}_{ii}(t)$, $\bar{z}_{ii}(t+1)$.
	\ENSURE $\bar{x}_i$.
	\STATE \textbf{Initialization:} \label{alg:correctionintial}
	$\epsilon_i(t_0)=\bigl[\frac{n^{\kappa-1}}{\bar{z}_{ii}(t_0+1) }-1\bigr]y_i$. 
	\STATE Receive data from neighbors: $\mathcal{Q}(x_j(t)+\epsilon_j(t))$, $j\in \mathcal{N}_i$.
	\STATE Update the state $x_i(t)$ via \eqref{eq:algorithm}.
	\STATE Compute the correction: \label{alg:correction} \\[1mm] 
	$\epsilon_i(t+1)=n^{\kappa-1}y_i\frac{\bar{z}_{ii}(t+1)-\bar{z}_{ii}(t+2)}{\bar{z}_{ii}(t+1)\bar{z}_{ii}(t+2)}$. \label{alg:epsilon}\\[1mm] 
	\IF{$t\geq t_0$}
	\STATE $K\triangleq t-t_0$.
	\STATE Update the average $\bar{x}_i(K)$:\\[1mm] \label{alg:updatex_i}
	$\bar{x}_i(K+1)=\frac{K}{K+1} \bar{x}_i(K)+\frac{1}{K+1} x_i(t)$. \\[1mm]
	\ENDIF
\end{algorithmic}
\end{algorithm}

\subsection{Summary of the algorithm}
	At each iteration, the proposed distributed estimation algorithm with quantized data  is composed of Algorithm~\ref{alg:estlefteig} and Algorithm~\ref{alg:estimation}. In the algorithm, we use an increasing window size $t-k_0$ (resp. $t-t_0$)  for the averaging process. A fixed window size $K$ can also be adopted according to what level of the convergence performance is needed. This can be inferred from the theoretical results in Section~\ref{sec:convergence_analysis}.

	We remark that the adjustment of the initial values in line~\ref{alg:initialz} of Algorithm~\ref{alg:estlefteig} has another consequence. It is known that $0<\omega_i<1$, $\forall i$, by Lemma~\ref{lem:matrixproperty} and some $\omega_i$'s are rather close to 0 for certain topologies. It is then probable that zeros would occur in the denominators of $\epsilon_i(t)$ during the quantization process, which makes the implementation of \eqref{eq:correction} meaningless. Increasing the initial values from 1 to $n^{\kappa}$ is meant to tackle this concern. Our simulation results validate this consideration.

 	We also emphasize that no further buffer is needed to store the previous states $\bar{\mathbf{z}}_i(K)$ and $\bar{x}_i(K)$ (see line \ref{alg:updatez_i} of Algorithm~\ref{alg:estlefteig} and line \ref{alg:updatex_i} of Algorithm~\ref{alg:estimation} for their recursive implementations). Further, the starting points $k_0,t_0$ contributes little to the rate of convergence of the algorithm in the long run.
 	But they do have an effect on the transient behaviors at the first few steps if  not appropriately designated.

    Finally, in order to deal with directed communication links, we introduce the left eigenvector estimation stage (Stage 1). However, there is no free lunch. The price we have to pay for the generality and performance of the algorithm is the increasing memory size at Stage 1, which is of the order $O(n)$. This limits its scalability  for large-scale sensor networks. A more efficient algorithm deserves further investigation. 

\section{Convergence analysis of the proposed averaging based algorithm} 
\label{sec:convergence_analysis}
	In this section, we first present the convergence  results for  the estimation algorithm of the left eigenvector $\boldsymbol{\omega}$, based on which the convergence analysis of the proposed averaging based algorithm is given. For notational simplicity, we assume that $k_0=t_0=0$ for the subsequent analysis.

\subsection{Convergence analysis of Algorithm~\ref{alg:estlefteig}}

	Write $\mathcal{Q}(\mathbf{z}_i(t))=\mathbf{z}_i(t)+\mathbf{u}_i(t)$, where $\mathbf{u}_i(t)$ is the quantization error with zero mean and $\mathbb{E}\{\|\mathbf{u}_i(t)\|^2\}\leq n\Delta^2/4$ in view of \eqref{eq:varianncequantierror}.  Let $\mathbf{U}(t)\triangleq [\mathbf{u}_1(t),\mathbf{u}_2(t),\dots, \mathbf{u}_n(t)]^T$, then we can write \eqref{eq:estimateeigenvec} in a compact form $\mathbf{Z}(t+1)=\mathbf{P} \mathbf{Z}(t)-\alpha \mathbf{L}\mathbf{U}(t)$ with $\mathbf{Z}(0)=\mathbf{I}$.
	Hence it can be derived  that
\begin{equation}\label{eq:barZ}
	\bar{\mathbf{Z}}(K)=\frac{1}{K}\sum_{k=1}^{K}\left(\mathbf{P}^k-\alpha\sum_{s=0}^{k-1}\mathbf{P}^{k-s-1} \mathbf{L} \mathbf{U}(s)\right).
\end{equation}

	Define the estimation error as $\mathbf{e}_{\bar{\mathbf{Z}}}(K)\triangleq \bar{\mathbf{Z}}(K)-\mathbf{1} \boldsymbol{\omega}^T$. Recall that $\mathbf{L} \mathbf{1}=\boldsymbol{\omega}^T \mathbf{L}=\mathbf{0}$, it is easy to verify that $\mathbf{P}^k-\mathbf{1}\boldsymbol{\omega}^T=\mathbf{Q}^k$ and $\mathbf{P}^k \mathbf{L}=\mathbf{Q}^k \mathbf{L}$, $\forall k\in \mathbb{Z}_{\geq 1}$. This together with \eqref{eq:barZ} implies
\begin{equation}\label{eq:eZ}
	\mathbf{e}_{\bar{\mathbf{Z}}}(K)=\frac{1}{K}\sum_{k=1}^{K} \left(\mathbf{Q}^k-\alpha
	\sum_{s=0}^{k-1}\mathbf{Q}^{k-s-1} \mathbf{L} \mathbf{U}(s)\right).
\end{equation}
	By Lemma \ref{lem:parameter}, we have $\rho(\mathbf{Q})<1$, which implies that $\mathbf{I}-\mathbf{Q}$ is nonsingular. 
	Moreover, by interchanging the order of summation, we can obtain
\[
	\sum_{k=1}^{K}\sum_{s=0}^{k-1}\mathbf{Q}^{k-s-1} \mathbf{L} \mathbf{U}(s)=\sum_{k=0}^{K-1}\sum_{s=0}^{K-k-1} \mathbf{Q}^s \mathbf{L}\mathbf{U}(k).
\]
	It thus follows from \eqref{eq:eZ} that
\begin{equation}\label{eq:eZeigenvec}
	\mathbf{e}_{\bar{\mathbf{Z}}}(K)=\frac{1}{K}\tilde{\mathbf{Q}}(\mathbf{I}-\mathbf{Q}^K)-\frac{\alpha}{K}\sum_{k=0}^{K-1}\mathbf{W}_K(k)\tilde{\mathbf{L}}\mathbf{U}(k),
\end{equation}
	where $\tilde{\mathbf{Q}}\triangleq (\mathbf{I}-\mathbf{Q})^{-1}\mathbf{Q}$, $\tilde{\mathbf{L}}\triangleq (\mathbf{I}-\mathbf{Q})^{-1}\mathbf{L}$, and $\mathbf{W}_K(k)\triangleq \mathbf{I}-\mathbf{Q}^{K-k}$, for $0\leq k\leq K-1$.

\subsubsection{Mean square performance} 
	Let $\mathbf{D}(t)\triangleq \mathbb{E}\{\mathbf{U}(t)\mathbf{U}^T(t)\}$. By Assumption~2, we can decompose it as $\mathbf{D}(t)=\mathbf{F}(t)^2$, where $\mathbf{F}(t)\triangleq \text{diag}\bigl\{\sqrt{\mathbb{E}\{\|\mathbf{u}_1(t)\|^2\}}, \dots,\sqrt{\mathbb{E}\{\|\mathbf{u}_n(t)\|^2\}}\bigr\}$. 
	Invoking \eqref{eq:eZeigenvec} and Assumption 2 on $\{\mathbf{U}(t)\}_{t\geq 0}$ implies 
\begin{equation}\label{eq:eZmeansquarebound}
\begin{split}
	\mathbb{E}\bigl\{\|\mathbf{e}_{\bar{\mathbf{Z}}}(K)\|_F^2\bigr\}&=\frac{1}{K^2}\bigl\|\tilde{\mathbf{Q}}(\mathbf{I}-\mathbf{Q}^K)\bigr\|_F^2\\
	&\relphantom{=}{}+\frac{\alpha^2}{K^2}\sum_{k=0}^{K-1} \bigl\|\mathbf{W}_K(k)\tilde{\mathbf{L}}\mathbf{F}(k)\bigr\|_F^2.
\end{split}
\end{equation}

	We have the following result regarding the mean square convergence of $\mathbf{e}_{\bar{\mathbf{Z}}}(K)$.
\begin{theorem}\label{thm:MSconvergence}
	Under Assumptions 1 and 2, $\bar{\mathbf{Z}}(K)$ converges in mean square to $\mathbf{1}\boldsymbol{\omega}^T$ as $K\to \infty$. Moreover, for large $K$, the mean square deviation is approximately given by 
\begin{equation}\label{eq:MSdeviation}
	\mathbb{E}\left\{\|\mathbf{e}_{\bar{\mathbf{Z}}}(K)\|_F^2\right\}\leq \frac{n\nu^2}{4} \frac{1}{K},
\end{equation}
	where $\nu\triangleq \alpha\sqrt{n+2}\Delta\|\tilde{\mathbf{L}}\|_2$.
\end{theorem}

\begin{proof}
	See Appendix~\ref{app:MSconvergence}.
\end{proof}

	Theorem~\ref{thm:MSconvergence} demonstrates that the averaging based method has a universal convergence rate of $O(K^{-1})$, independent of the network topology. This is a distinctive feature of the proposed algorithm from the standard consensus algorithm \cite{RenBea05}. The possible effect of the network topology only lies in the rate coefficient $\lim_{K\to\infty} K \mathbb{E}\{\|\mathbf{e}_{\bar{\mathbf{Z}}}(K)\|_F^2\}$.
	In fact, the upper bound \eqref{eq:MSdeviation} gives a rough estimate of the rate coefficient, i.e., $ n(n+2) \alpha^2\Delta^2\|(\alpha \mathbf{L}+\mathbf{1}\boldsymbol{\omega}^T)^{-1}\mathbf{L}\|_2^2 /4$, 
	which depends on the parameter $\alpha$, the network topology through $n$, $\mathbf{L}$ and $\boldsymbol{\omega}$,  and the quantization scheme through $\Delta$. We note that a similar form of the rate of convergence is established  in \cite{FangLi10}.

\subsubsection{Almost sure performance} 
	It follows from Lemma~\ref{lem:parameter} that 
	$\|\tilde{\mathbf{Q}}(\mathbf{I}-\mathbf{Q}^t)\|_F\leq \|\tilde{\mathbf{Q}}\|_2\|\mathbf{I}-\mathbf{Q}^t\|_F \leq \|\tilde{\mathbf{Q}}\|_2(\sqrt{n+2}+nc_\mathbf{Q} t^{q-1}\rho^{t}(\mathbf{Q}))$. This together with \eqref{eq:eZeigenvec} gives 
\begin{multline}\label{eq:eZasbound}
	\|\mathbf{e}_{\bar{\mathbf{Z}}}(K)\|_F\leq \frac{\sqrt{n+2}\|\tilde{\mathbf{Q}}\|_2}{K}+ nc_\mathbf{Q} \|\tilde{\mathbf{Q}}\|_2K^{q-2}\rho^{K}(\mathbf{Q})\\+\frac{\alpha}{K}\left\|\sum_{k=0}^{K-1}\mathbf{W}_K(k)\tilde{\mathbf{L}}\mathbf{U}(k)\right\|_F.
\end{multline}
	Obviously, the first two terms of the RHS of \eqref{eq:eZasbound} tend to zero as $K\to \infty$. The third term is in the form of weighted sum of random matrices. The law of the iterated logarithm for independent random variables \cite[Chap.8]{Gut13} motivates us to provide similar quantitative bounds on the rate of convergence of the third term. To this end, we define
\begin{equation}\label{eq:rU}
    r_{K}^\mathbf{U}\triangleq \max_{i} \lambda_{\max}\left(\sum_{k=0}^{K-1} \text{Cov}(\mathbf{u}_i(k))\right).
\end{equation}

\begin{theorem}\label{thm:eZas}
	Under Assumptions 1 and 2, for all large $K$,

	i) if $\sup_{K\geq 1}r_K^\mathbf{U}<\infty$, then there exists a constant $c_\mathbf{U}>0$ such that $\max_{i}\|\sum_{k=0}^{K} \mathbf{u}_i(k)\|\leq c_\mathbf{U}$ a.s. and 
\begin{equation}\label{eq:asboundboundednoise}
	\|\mathbf{e}_{\bar{\mathbf{Z}}}(K)\|_F\leq \mu\frac{1}{K}\ \ \mbox{a.s.},
\end{equation}
	where $\mu\triangleq \sqrt{n+2}\|\tilde{\mathbf{Q}}\|_2+\alpha n (nc_\mathbf{Q}c_{\mathbf{Q}}'\Delta +c_\mathbf{U})\|\tilde{\mathbf{L}}\|_2 $, and
\[
	c_{\mathbf{Q}}'\triangleq \begin{cases}
    \frac{\rho(\mathbf{Q})}{1-\rho(\mathbf{Q})}, &q=1,\\
	(\frac{1-q}{e \log \rho(\mathbf{Q})})^{q-1}
	+\sum_{j=0}^{q-1} \frac{(q-1)! \rho(\mathbf{Q})}{j!(-\log \rho(\mathbf{Q}))^{q-j}}, &q> 1.
	\end{cases}
\]

	ii) if $\lim_{K\to \infty} r_K^\mathbf{U}=\infty$, then 
\begin{equation}\label{eq:asboundboundednoise2}
	\|\mathbf{e}_{\bar{\mathbf{Z}}}(K)\|_F\leq \alpha n \|\tilde{\mathbf{L}}\|_2 \frac{\sqrt{2 r_K^\mathbf{U} \log\log r_K^\mathbf{U}}}{K}\ \ \mbox{a.s.}
\end{equation}
\end{theorem}
\begin{proof}
	See Appendix~\ref{app:eZas}.
\end{proof}

	By \eqref{eq:rU}, we can deduce that  
\begin{align*}
	r_K^\mathbf{U}&\leq \max_i \sum_{k=0}^{K-1} \lambda_{\max}(\text{Cov}(\mathbf{u}_i(k)))\\
	&\leq \max_i \sum_{k=0}^{K-1} \mathbb{E}\{\|\mathbf{u}_i(k)\|^2\}\leq \frac{n\Delta^2}{4} K,
\end{align*}
    where the second step follows from the relation that $\lambda_{\max}(\text{Cov}(\mathbf{u}_i(k)))\leq \mathbb{E}\{\|\mathbf{u}_i(k)\|^2\}$, and the last inequality is a direct consequence of \eqref{eq:boundquantierror}. Note that $\log K=o(\sqrt{K})$, Theorem~\ref{thm:eZas} thus reveals that $\lim_{K\to \infty}\mathbf{e}_{\bar{\mathbf{Z}}}(K)= 0$ a.s.. This means that the left eigenvector $\boldsymbol{\omega}$ can be asymptotically obtained at each node by using the running average technique, which establishes the convergence property of Algorithm~\ref{alg:estlefteig} in the almost sure sense.

	Theorem~\ref{thm:eZas} has another important implication. Actually, we have 
\begin{equation}\label{eq:barz-wi_errineq}
	\sum_{i=1}^n |\bar{z}_{ii}(K)-\omega_i|^2\leq \|\mathbf{e}_{\bar{\mathbf{Z}}}(K)\|_F^2,  \ \forall K\in \mathbb{Z}_{\geq 0}.
\end{equation}
	Hence, by Theorem~\ref{thm:eZas}, $\lim_{K\to \infty}\bar{z}_{ii}(K)=\omega_i$~a.s.. Moreover, by Lemma~\ref{lem:matrixproperty}, we know that  $\min_{i}\omega_i>0$.
	Thus, for all large $t$, one has $\bar{z}_{ii}(t)\geq \eta \omega_i$ a.s., $\forall i$. The above discussion leads to the following theorem.
\begin{theorem}\label{thm:welldefined}
	Let Assumptions 1 and 2 hold, then for any constant $0<\eta<1$, there exists $t_\eta\in \mathbb{Z}_{\geq 0}$ such that 
\begin{equation}\label{eq:barziibound}
	\min_i \frac{\bar{z}_{ii}(t)}{w_i} \geq \eta  \  \ \text{a.s.} ,\  \forall t\in \mathbb{Z}_{\geq t_\eta}.
\end{equation}
\end{theorem}

	Theorem~\ref{thm:welldefined} states that the correction term $\epsilon_i(t)$ in \eqref{eq:correction} is well-defined for large $t$. For the implementation of \eqref{eq:correction}, one may choose $t_0=t_\eta$ to trigger the estimation algorithm at the second stage. Actually, with the setup in Algorithm~\ref{alg:estlefteig}, it is possible to choose a much smaller $t_0$ (see the simulation results in Section~\ref{sec:simulation_results}). For clarity of presentation of the subsequent analysis,  we always assume that $\bar{z}_{ii}(t)\geq \eta w_i$, $\forall t\in \mathbb{Z}_{\geq t_0}$.

	\begin{table*}
		\caption{Upper bounds of $\|\mathbf{e}_{\bar{\mathbf{x}}}(K)\|$}
		\label{tab:upperbounderror}
		\begin{center}
			\begin{tabular}{l||l}
			\hline\hline
	        $\sup_K r_{K}^{\mathbf{v}}<\infty$, $\sup_K r_{K}^\mathbf{U}<\infty$ & $\sqrt{2}\mu (n\eta^2)^{-1}(2c_{\mathbf{Q},n}y'\|\tilde{\mathbf{Q}}\|_2+\sqrt{n}y'') \times K^{-1}\log K$\\
	        $\sup_K r_{K}^{\mathbf{v}}<\infty$, $r_{K}^\mathbf{U}\to \infty$ & $2\alpha \eta^{-2}(2c_{\mathbf{Q},n}y'\|\tilde{\mathbf{Q}}\|_2+\sqrt{n}y'')\|\tilde{\mathbf{L}}\|_2 \times K^{-1}\log K \sqrt{r_{K}^\mathbf{U} \log\log r_{K}^\mathbf{U}}$\\
	        $r_{K}^{\mathbf{v}}\to \infty$, $\sup_K r_{K}^\mathbf{U}<\infty$ &$\sqrt{2} (n\eta^2)^{-1}(\alpha n\eta^2\|\tilde{\mathbf{L}}\|_2+\mu (2c_{\mathbf{Q},n}y'\|\tilde{\mathbf{Q}}\|_2+\sqrt{n}y'')) \times K^{-1}\max\bigl\{\sqrt{r_K^{\mathbf{v}} \log\log r_K^{\mathbf{v}}},\ \log K\bigr\}$\\
	        $r_{K}^{\mathbf{v}}\to \infty$, $r_{K}^\mathbf{U}\to \infty$& $\sqrt{2}\alpha \eta^{-2}(1+\sqrt{2}(2c_{\mathbf{Q},n}y'\|\tilde{\mathbf{Q}}\|_2+\sqrt{n}y''))\|\tilde{\mathbf{L}}\|_2 \times K^{-1}\max\Bigl\{\sqrt{r_K^{\mathbf{v}} \log\log r_K^{\mathbf{v}}}, \ \sqrt{r_{K}^\mathbf{U} \log\log r_{K}^\mathbf{U}} \log K\Bigr\}$\\
			\hline\hline
			\end{tabular}
		\end{center}
	\end{table*}
\subsection{Convergence analysis of Algorithm~\ref{alg:estimation} } 
\label{sub:convergence_analysis}

	Write $\mathcal{\mathbf{Q}}(x_i(t))=x_i(t)+v_i(t) $, $\forall i$, where $v_i(t)$ is the quantization error satisfying \eqref{eq:varianncequantierror} and \eqref{eq:boundquantierror}. Stack $x_i(t) $, $\epsilon_i(t) $ and $v_i(t) $ into column vectors $\mathbf{x}(t)$, $\boldsymbol{\epsilon}(t)$ and $\mathbf{v}(t)$, respectively,  then we can rewrite \eqref{eq:algorithm} more compactly into
\begin{equation}\label{eq:compactalgorithm}
	\mathbf{x}(t+1)=\mathbf{P}(\mathbf{x}(t)+\boldsymbol{\epsilon}(t))-\alpha \mathbf{L}\mathbf{v}(t).
\end{equation}
	Hence the running average $\bar{\mathbf{x}}(K)$ of \eqref{eq:arithmeticmean} can be expressed as
\begin{equation}\label{eq:averagecompact}
	\begin{split}
	\bar{\mathbf{x}}(K)&=\frac{1}{K} \sum_{k=1}^{K}  \left(\mathbf{P}^k \mathbf{y}+\sum_{s=0}^{k-1}\mathbf{P}^{k-s-1}(\mathbf{P}\boldsymbol{\epsilon}(s)-\alpha \mathbf{L} \mathbf{v}(s))\right). 
	\end{split}
\end{equation}

	The next lemma provides the convergence properties of the correction term $\boldsymbol{\epsilon}(t)$ of \eqref{eq:correction}.

\begin{lemma}\label{lem:correctionconvergence}
	Let  Assumptions 1 and 2 hold, then  for each $0<\eta<1$, we have for all large $t$,
\[
	\mathbb{E}\{\|\boldsymbol{\epsilon}(t)\|^2\} \leq \frac{\nu^2 y'^2}{n\eta^4} \frac{1}{t},
\]
	and almost surely
\[
	\|\boldsymbol{\epsilon}(t)\| \leq \begin{cases}
	\frac{2 \sqrt{2}\mu y'}{n\eta^2}\frac{1}{t}, &\sup r_K^\mathbf{U}<\infty,\\
	\frac{4\alpha y'\|\tilde{\mathbf{L}}\|_2}{\eta^2}\frac{\sqrt{r_{t+1}^\mathbf{U} \log\log r_{t+1}^\mathbf{U}} }{t}, & \text{otherwise},
	\end{cases}
\]
	where $y'\triangleq \max_{i} \omega_i^{-2}|y_i|$.
\end{lemma}
\begin{proof}
	See Appendix~\ref{app:correctionconvergence}.
\end{proof}

	The compensation nature of  $\boldsymbol{\epsilon}(t)$ is demonstrated in the next lemma, which guarantees convergence of the weighted sum to the desired $\hat{\theta}$. To this end, we denote $e_{\mathbf{x}}(t)\triangleq \boldsymbol{\omega}^T \mathbf{x}(t)-\hat{\theta}$.
\begin{lemma}\label{lem:weightsum}
	Let Assumptions 1 and 2 hold, then $e_{\mathbf{x}}(t)$ is approximately bounded by 
\[
	\mathbb{E}\bigl\{e_{\mathbf{x}}^2(t)\bigr\}\leq \frac{\nu^2 y''^2}{4\eta^2} \frac{1}{t},
\]
	and  almost surely
\[
	|e_{\mathbf{x}}(t)|\leq \begin{cases}
	\frac{\mu y''}{\eta\sqrt{n}}\frac{1}{t}, & \sup r_K^\mathbf{U}<\infty,\\
	\frac{\sqrt{2n}\alpha y''\|\tilde{\mathbf{L}}\|_2}{\eta} \frac{\sqrt{r_{t}^\mathbf{U} \log\log r_{t}^\mathbf{U}} }{t}, &\text{otherwise},
	\end{cases}
\]
	where $y''\triangleq \max_{1\leq i\leq n} \omega_i^{-1} |y_i|$.
\end{lemma}
\begin{proof}
	See Appendix~\ref{app:MSweightsum}.
\end{proof}

	Based on Lemma~\ref{lem:weightsum}, we can decompose the estimation error $\mathbf{e}_{\bar{ \mathbf{x}}}(K)\triangleq \bar{\mathbf{x}}(K)-\hat{\theta} \mathbf{1}$ into two parts:  $\mathbf{e}_{\bar{ \mathbf{x}}}(K)=e_{\mathbf{x}}(K)\mathbf{1}+\mathbf{e}_{\bar{\mathbf{x}},\mathbf{x}}(K)$, where $\mathbf{e}_{\bar{\mathbf{x}},\mathbf{x}}(K)\triangleq \bar{\mathbf{x}}(K)-\boldsymbol{\omega}^T \mathbf{x}(K)\mathbf{1}$.  In the following, it suffices to  provide an upper bound of $\|\mathbf{e}_{\bar{\mathbf{x}},\mathbf{x}}(K)\|$. In fact, similar to \eqref{eq:eZeigenvec}, one can obtain from \eqref{eq:compensate} and \eqref{eq:averagecompact} that
\begin{align}\label{eq:tilde_eK}
\mathbf{e}_{\bar{\mathbf{x}},\mathbf{x}}(K)
	&=\frac{1}{K}\Biggl(\underset{\mathcal{I}_1}{\underbrace{\tilde{\mathbf{Q}}(\mathbf{I}-\mathbf{Q}^K)\mathbf{y}-\alpha \sum_{k=0}^{K-1}\mathbf{W}_K(k)\tilde{\mathbf{L}}\mathbf{v}(k)}}\Biggr.\notag\\
	&\relphantom{=}{}+\underset{\mathcal{I}_2}{\underbrace{\tilde{\mathbf{Q}}\sum_{k=0}^{K-1}\mathbf{W}_K(k)\boldsymbol{\epsilon}(k)}}+\underset{\mathcal{I}_3}{\underbrace{\frac{1}{n}\sum_{k=1}^{K}\mathbf{1} \mathbf{x}^T(0)\boldsymbol{\varepsilon}_{K}(k)}}\Biggr),
\end{align}
	where $\boldsymbol{\varepsilon}_K(k)=[\varepsilon_{1K}(k), \dots, \varepsilon_{nK}(k)]^T$ with the $i$-th entry being $\varepsilon_{iK}(k)=\omega_i[1/\bar{z}_{ii}(k)-1/\bar{z}_{ii}(K)]$.
\subsubsection{Mean square performance}
	We have the next result regarding the mean square convergence of $\mathbf{e}_{\bar{\mathbf{x}},\mathbf{x}}(K)$. 
\begin{lemma}\label{lem:MSbarxx_eK}
	Under Assumptions 1 and 2,  we have 
\[
	\mathbb{E}\left\{\|\mathbf{e}_{\bar{\mathbf{x}},\mathbf{x}}(K) \|^2\right\}\leq \frac{3\nu^2 (ny''^2+2c_{\mathbf{Q},n}^2y'^2\|\tilde{\mathbf{Q}}\|_2^2)}{2n\eta^4}\frac{\log K}{K}.
\]
\end{lemma} 
\begin{proof}
	See Appendix~\ref{app:MStilde_eK}.
\end{proof}

	Note that $\mathbb{E}\left\{\|\mathbf{e}_{\bar{\mathbf{x}}}(K) \|^2\right\}\leq 2 \mathbb{E}\{\|e_{\mathbf{x}}(K)\mathbf{1}\|^2\}+2 \mathbb{E}\{\|\mathbf{e}_{\bar{\mathbf{x}},\mathbf{x}}(K)\|^2\}$, we immediately have the next result of $\mathbb{E}\left\{\|\mathbf{e}_{\bar{\mathbf{x}}}(K) \|^2\right\}$ based on Lemmas~\ref{lem:weightsum} and \ref{lem:MSbarxx_eK}.
\begin{theorem}\label{thm:MSe_Kbound}
	Let Assumptions 1 and 2 hold, then at each node $i$, the running average $\bar{x}_i(K)$ converges to the centralized estimate $\hat{\theta}$ in mean square sense. Moreover, the mean square deviation is approximately bounded by 
\begin{equation*}
	\mathbb{E}\left\{\|\mathbf{e}_{\bar{\mathbf{x}}}(K) \|^2\right\}\leq  \frac{3\nu^2 (ny''^2+2c_{\mathbf{Q},n}^2y'^2\|\tilde{\mathbf{Q}}\|_2^2)}{n \eta^4 }\frac{\log K}{K}.
\end{equation*}
\end{theorem}

\subsubsection{Almost sure performance}
	Before we move on to the almost sure analysis of $\|\bar{\mathbf{x}}(K)\|$, we introduce a similar function as in \eqref{eq:rU} 
\begin{equation}\label{eq:rv}
    r_K^{\mathbf{v}}\triangleq  \lambda_{\max}\left(\sum_{t=0}^{K-1} \text{Cov}(\mathbf{v}(t))\right).
\end{equation}
	
	Analogue to Lemma~\ref{lem:MSbarxx_eK} and Theorem~\ref{thm:MSe_Kbound}, we have the following result regarding the almost sure performance of $\|\bar{\mathbf{x}}(K)\|$. 
\begin{theorem}\label{thm:ase_Kbound}
	Let Assumptions 1 and 2 hold, then at each  node $i$, the running average $\bar{x}_i(K)$ converges to the centralized estimate $\hat{\theta}$ almost surely. Moreover,  for large $t$, the approximate upper bounds of $\|\mathbf{e}_{\bar{\mathbf{x}}}(K)\|$ are summarized in Table~\ref{tab:upperbounderror}.
\end{theorem} 

\begin{proof}
	See Appendix~\ref{app:ase_Kbound}.
\end{proof}

	From Theorems~\ref{thm:MSe_Kbound} and \ref{thm:ase_Kbound}, we can see that the starting point $t_0$ contributes little to the rate of convergence of the proposed algorithm, since $\log(t_0+K)\approx \log K$, for large $K$. This means that we can start the running averages $\bar{\mathbf{Z}}(K)$ and $\bar{\mathbf{x}}(K)$ at any time during the iteration. This is exactly what we have done in Algorithms~\ref{alg:estlefteig} and \ref{alg:estimation} by introducing the starting points $k_0,t_0$ for the averaging processes.  

\begin{remark}
   Existing results of consensus algorithms over undirected networks show that the sample mean estimate can be achieved in the mean square sense in the presence of quantization errors only if the quantization error variance at each node converges to 0 \cite{FangLi09,ThanKokPuFros13}. However, with the running average technique, the proposed algorithm is proven to be convergent to the sample mean estimate both in the mean square and almost sure senses without such restrictive requirement. This validates the advantage of the running average technique in dealing with the random quantization errors for distributed estimation problems over sensor networks.
\end{remark}

\section{Simulation results} 
\label{sec:simulation_results}
	In this section, we provide some simulation results to validate the theoretical results given in the previous section.
\begin{figure}[!t]
	\begin{center}
		\includegraphics[width=2.5in]{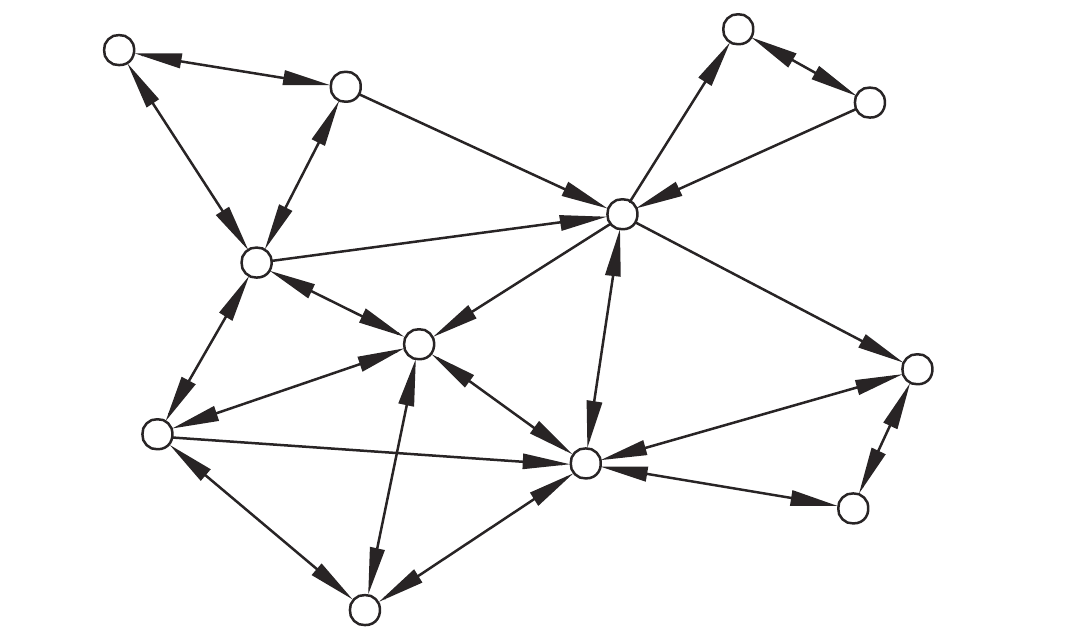}
	\end{center}
	\caption{A sensor network with 12 nodes modeled as a directed graph.}	
	\label{fig:topology}
\end{figure}
\begin{figure}[!t]
	\begin{center}
		\includegraphics[width=2.5in]{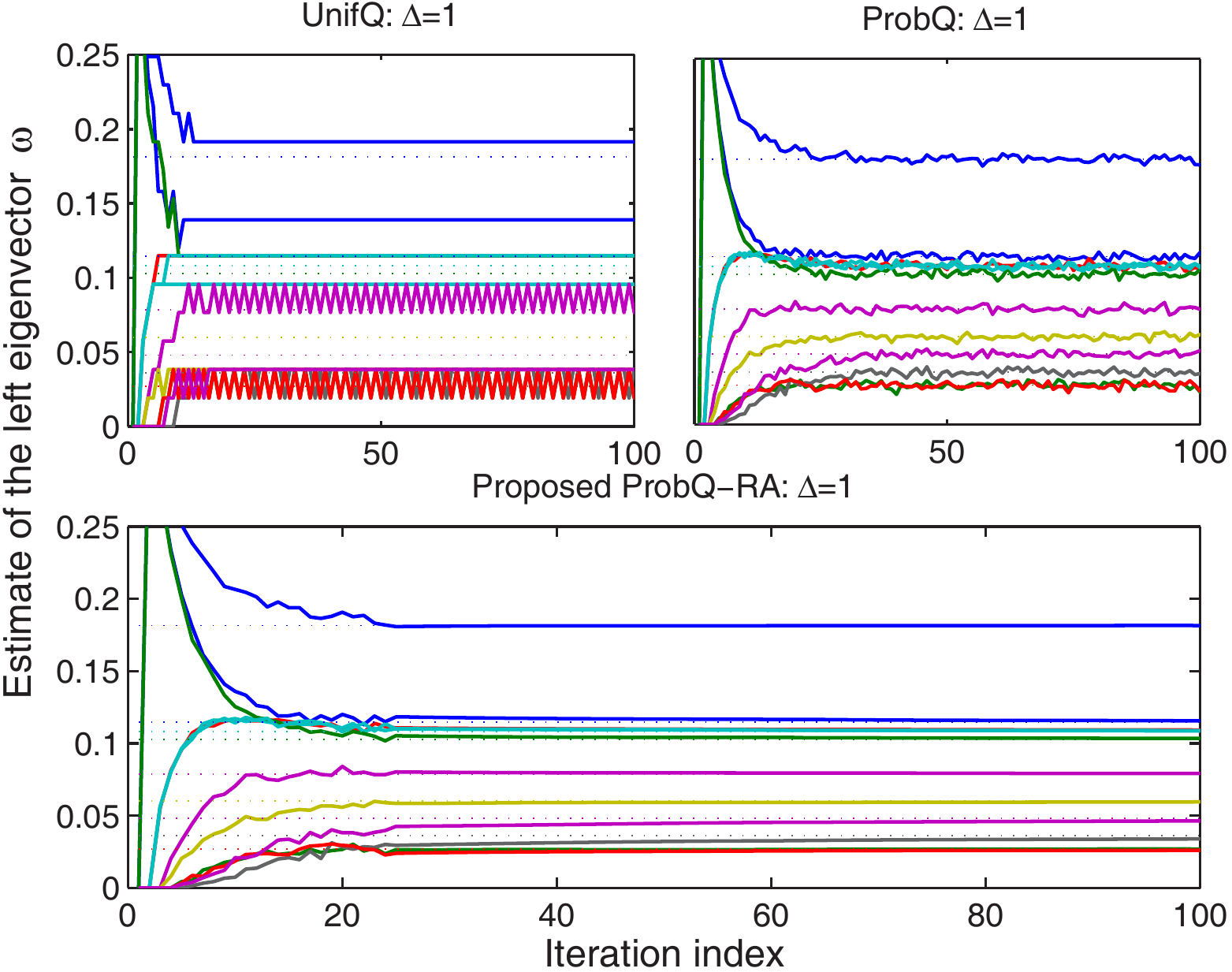}
	\end{center}
	\caption{Estimate of the left eigenvector $\boldsymbol{\omega}$  of one node for $\Delta=1$.}
	\label{fig:leigvec}
\end{figure}

	Consider a sensor network with 12 nodes to monitor an unknown parameter $\theta=2$. The directed communication topology is shown in Fig.~\ref{fig:topology}. Each node makes the measurement with $y_i=\theta+ n_i$, where $n_i$ is the white Gaussian noise with zero mean and unit variance. As an illustration, we choose the Metropolis-type weight $a_{ij}=(1+d_i)^{-1}$, if $j\in \mathcal{N}_i$ and 0, otherwise. In this case,  $\alpha=1$ is sufficient for both Lemmas~\ref{lem:matrixproperty} and \ref{lem:parameter}. 
	For each implementation of the proposed algorithm, the initial state $x_i(0)$ is randomly chosen from the interval $[y_i-1,y_i+1]$, $\forall i$. 

	In the following simulations, both the deterministic uniform quantization (UnifQ) \cite{CarFagFraZam10,GerGray92} and probabilistic quantization (ProbQ) \cite{AysCoatRabb08,CarFagFraZam10} are considered and compared. The proposed averaging based algorithm is denoted as ProbQ-RA.
	Simulation results are averages over 100 independent runs. 

\begin{figure}[!t]
	\begin{center}
		\subfloat[]{\includegraphics[width=1.75in]{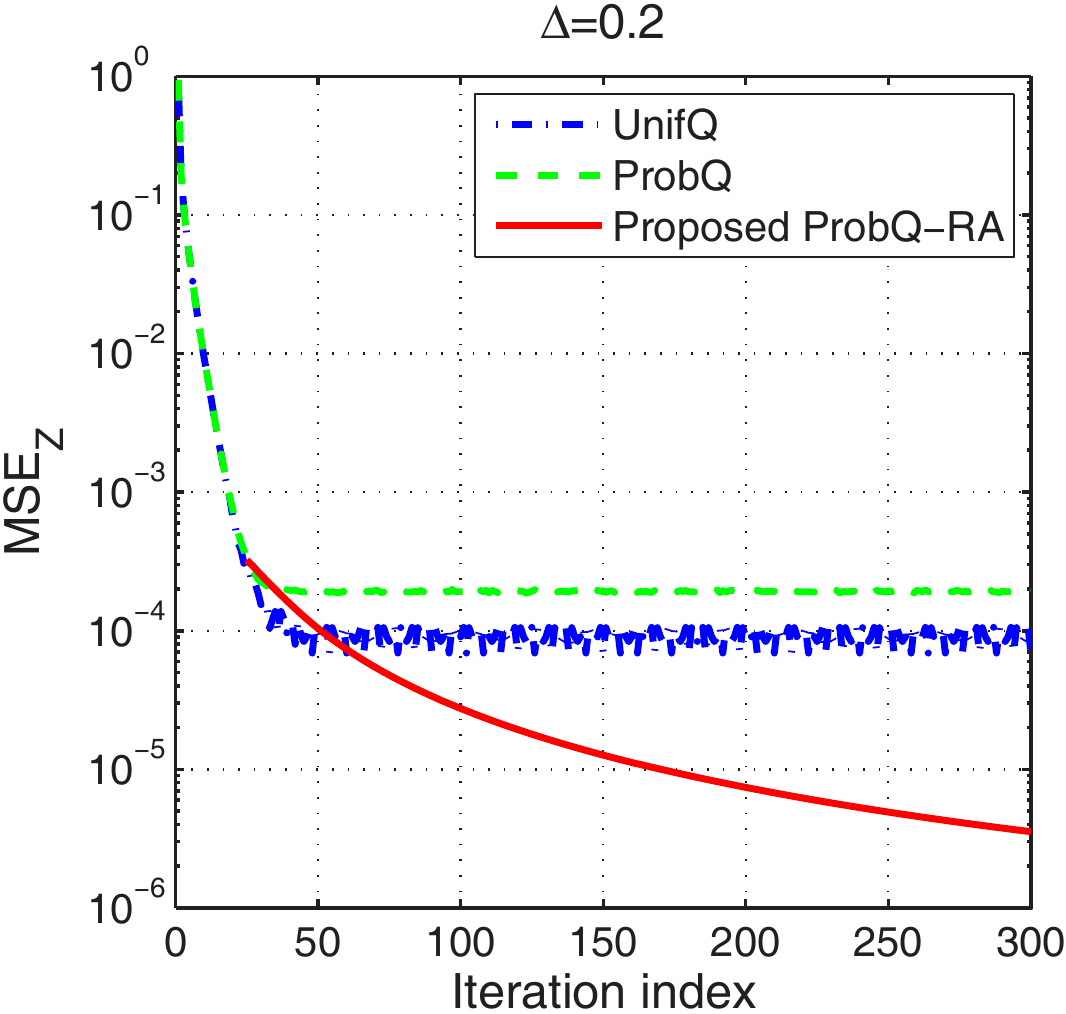}}
		\subfloat[]{\includegraphics[width=1.75in]{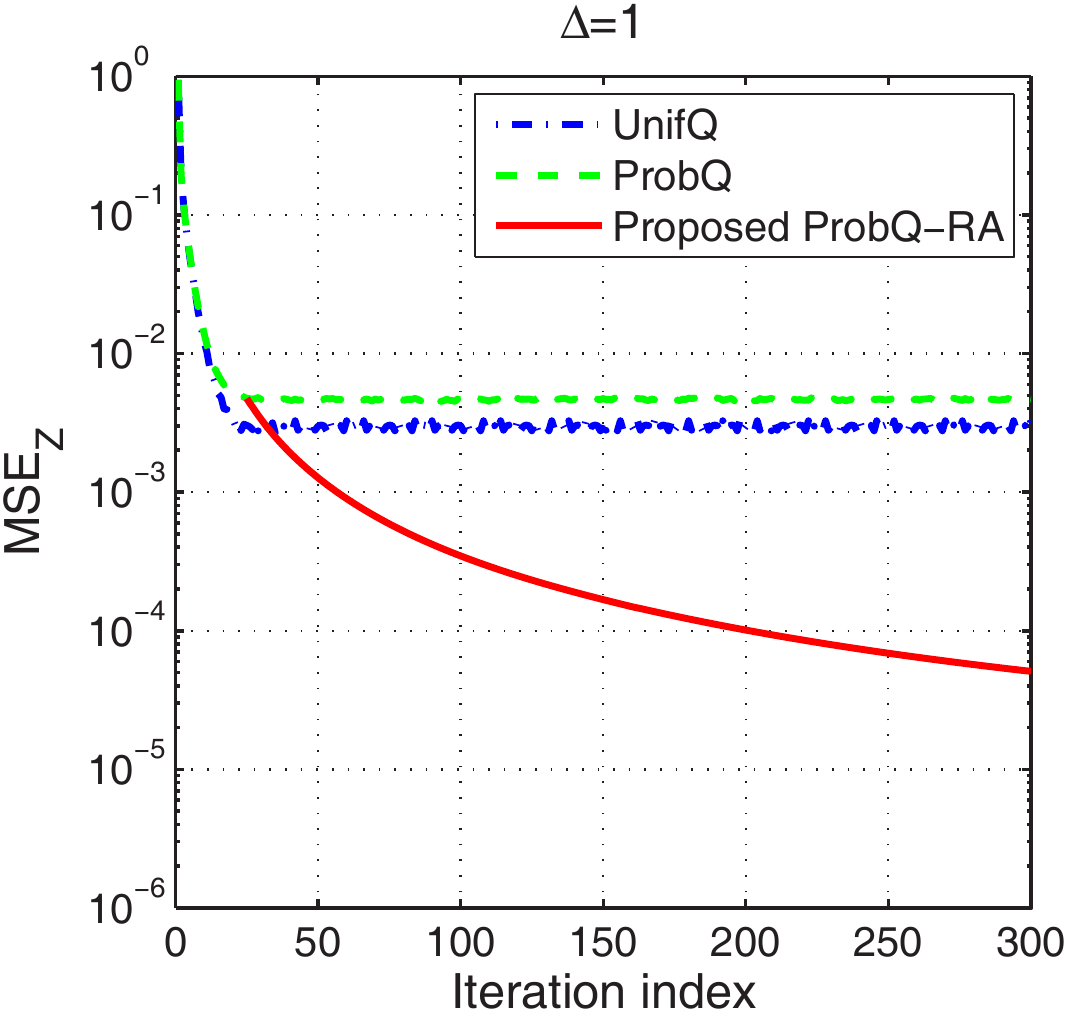}}\\
		\subfloat[]{\includegraphics[width=1.75in]{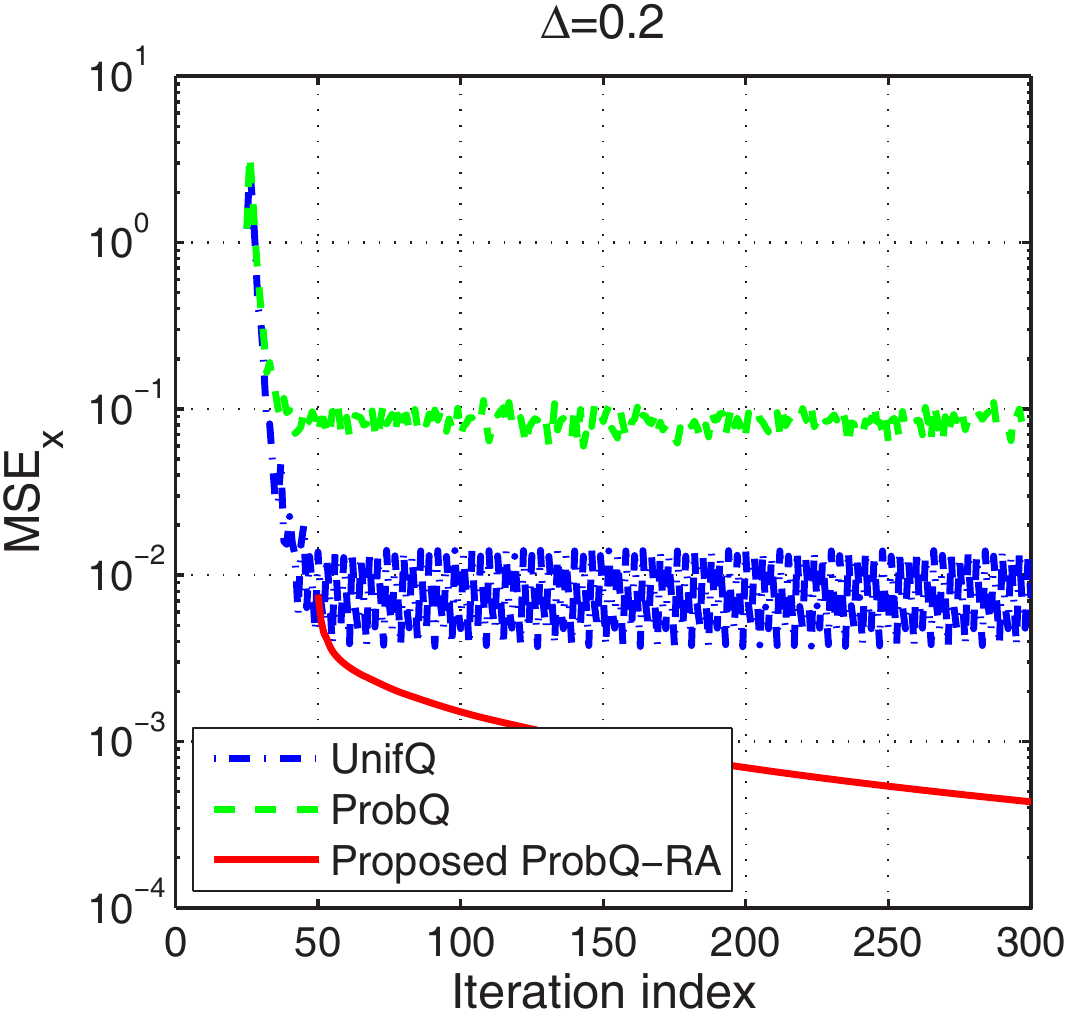}}
		\subfloat[]{\includegraphics[width=1.75in]{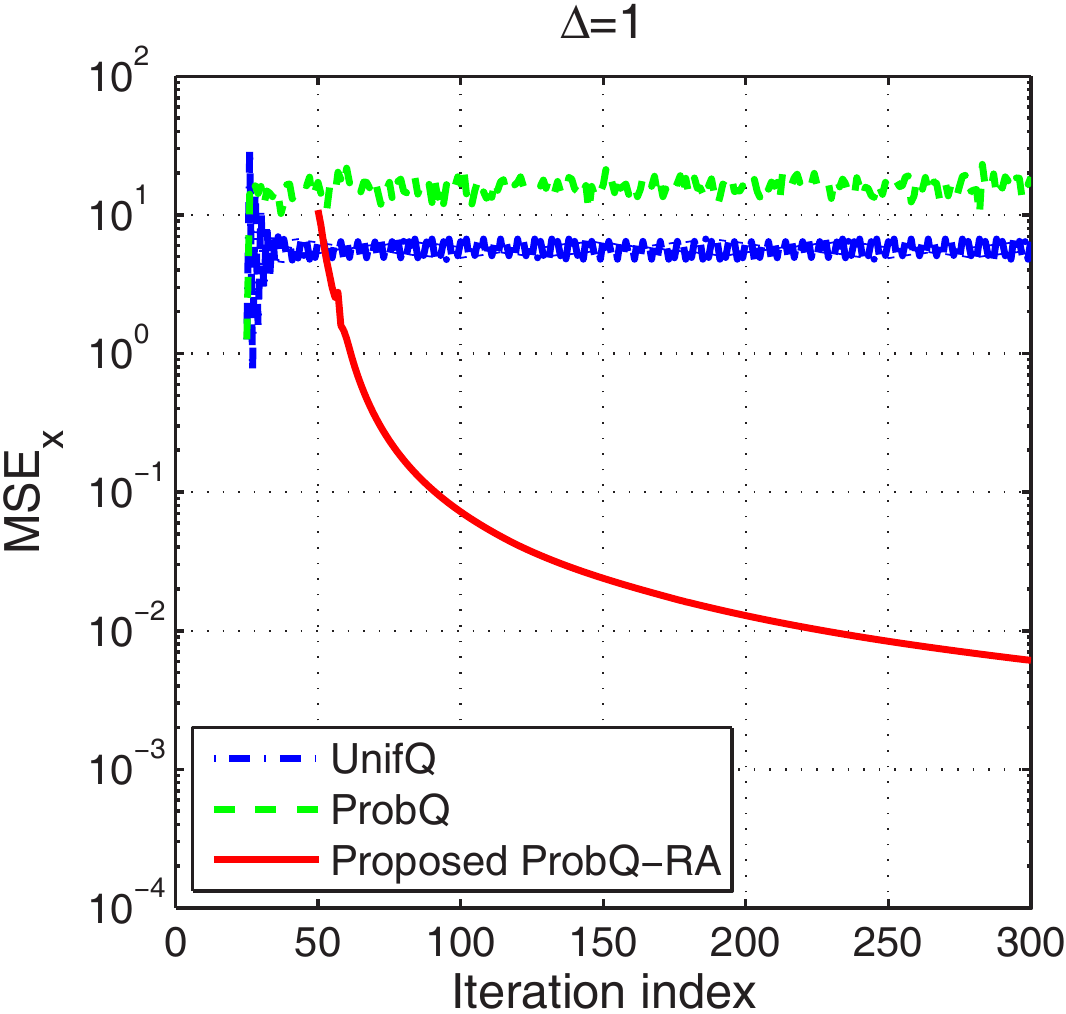}}
	\end{center}
	\caption{Comparison of  the mean square errors of UnifQ, ProbQ and ProbQ-RA: MSE${}_{\mathbf{Z}}$ ((a) and (b)), and MSE${}_{\mathbf{x}}$ ((c) and (d)) with respect to $\Delta\in \{0.2,1\}$.}	
	\label{fig:MSEzx}
\end{figure} 

\subsection{Comparison of the deterministic and probabilistic quantization}
	First, we simulate the eigenvector estimation algorithm of  Algorithm~\ref{alg:estlefteig}. Here, $\kappa=1.15$ and the starting point is taken as $k_0=25$.
	Fig.~\ref{fig:leigvec} depicts the estimate of the left eigenvector $\boldsymbol{\omega}$ at one node for $\Delta=1$.  In Fig.~\ref{fig:leigvec}, we use the original state of ProbQ as the estimate of $\boldsymbol{\omega}$ in the first 25 steps. From the results, we observe that steady residues occur for UnifQ, and there are fluctuations for ProbQ. While for the proposed ProbQ-RA, the running average has an obvious smoothing effect, where the randomness of ProbQ is smeared out. The performance of ProbQ-RA is rather satisfactory  compared with the large residues observed in both UnifQ and ProbQ.

    To quantify the performances, we use the average of the mean square error as an indicator, for Algorithm~\ref{alg:estlefteig}, we define
    \[
    \text{MSE}_{\mathbf{Z}}=\frac{1}{n}\sum_{i=1}^n \|\mathbf{z}_i(t)-\boldsymbol{\omega}\|^2, \
    \text{MSE}_{\bar{\mathbf{Z}}}=\frac{1}{n}\sum_{i=1}^n \|\bar{\mathbf{z}}_i(K)-\boldsymbol{\omega}\|^2,
    \] 
    while for Algorithm~\ref{alg:estimation},
    we let
    \[
    \text{MSE}_{\mathbf{x}}=\frac{1}{n} \sum_{i=1}^n (x_i(t)-\hat{\theta})^2, \
    \text{MSE}_{\bar{\mathbf{x}}}=\frac{1}{n} \sum_{i=1}^n (\bar{x}_i(K)-\hat{\theta})^2.
    \]
    The starting point used in  Algorithm~\ref{alg:estimation} is set as $t_0=25$. The results are shown in Fig.~\ref{fig:MSEzx}. It can be seen that the proposed ProbQ-RA outperforms UnifQ and ProbQ in both cases with the quantization resolutions $\Delta=0.2$ and 1. The performances of UnifQ and ProbQ are acceptable for the estimates of the left eigenvector $\boldsymbol{\omega}$ in both cases (see Fig.~\ref{fig:MSEzx}(a) and (b)). However, with the errors accumulated from the first stage to the second stage, they degrade significantly for lower quantization resolutions, e.g., $\Delta=1$ (see Fig.~\ref{fig:MSEzx}(c) and (d)). Compared with UnifQ and ProbQ, the proposed ProbQ-RA  degrades quite smoothly. There is only a modest increase of MSE with decreasing quantization resolution, i.e., increasing $\Delta$ from 0.2 to 1. These results indicate that  the averaging technique can improve the accuracy of the estimates especially for the case of low quantization resolutions, 
    where its smoothing effect contributes much to the improvement.
 
 \begin{figure}
	\begin{center}
		\subfloat[]{\includegraphics[width=1.75in]{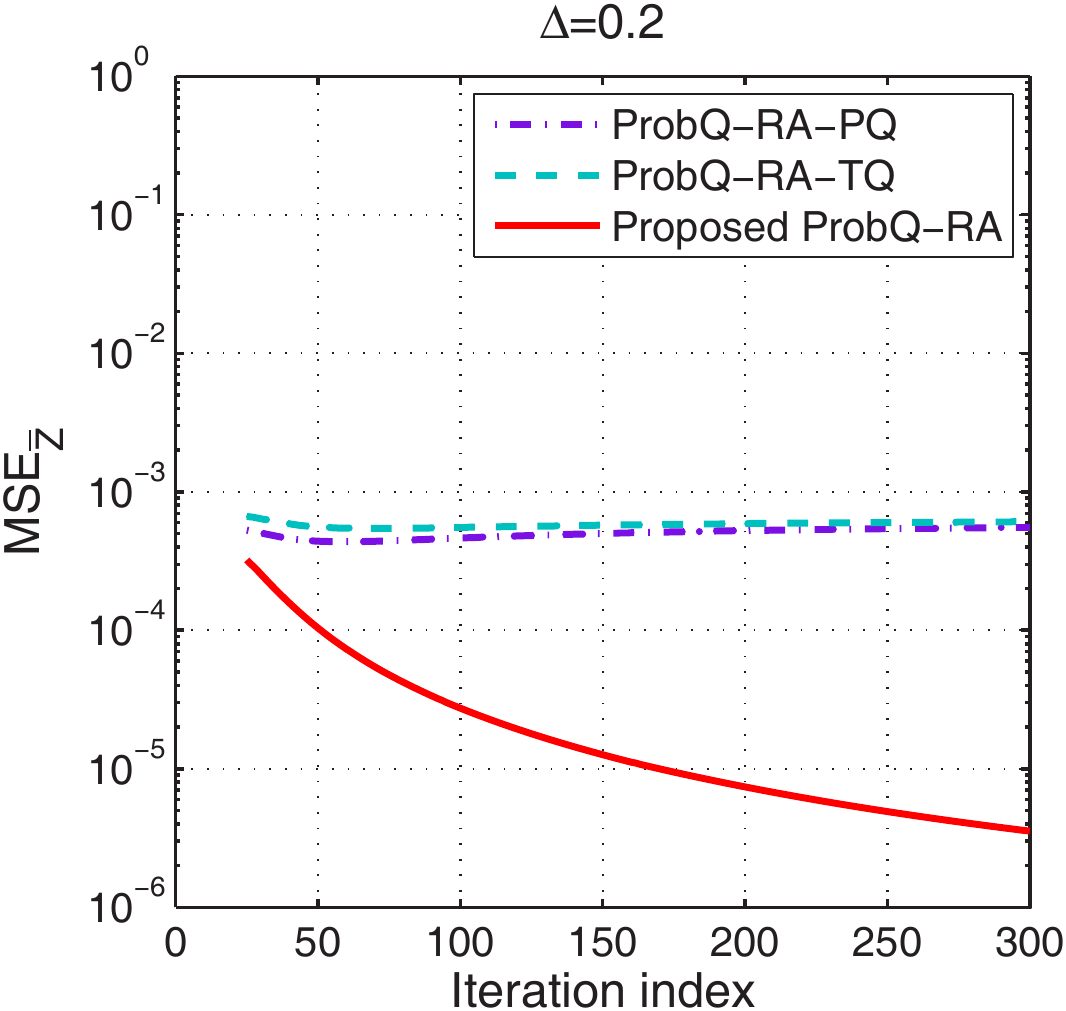}}
		\subfloat[]{\includegraphics[width=1.75in]{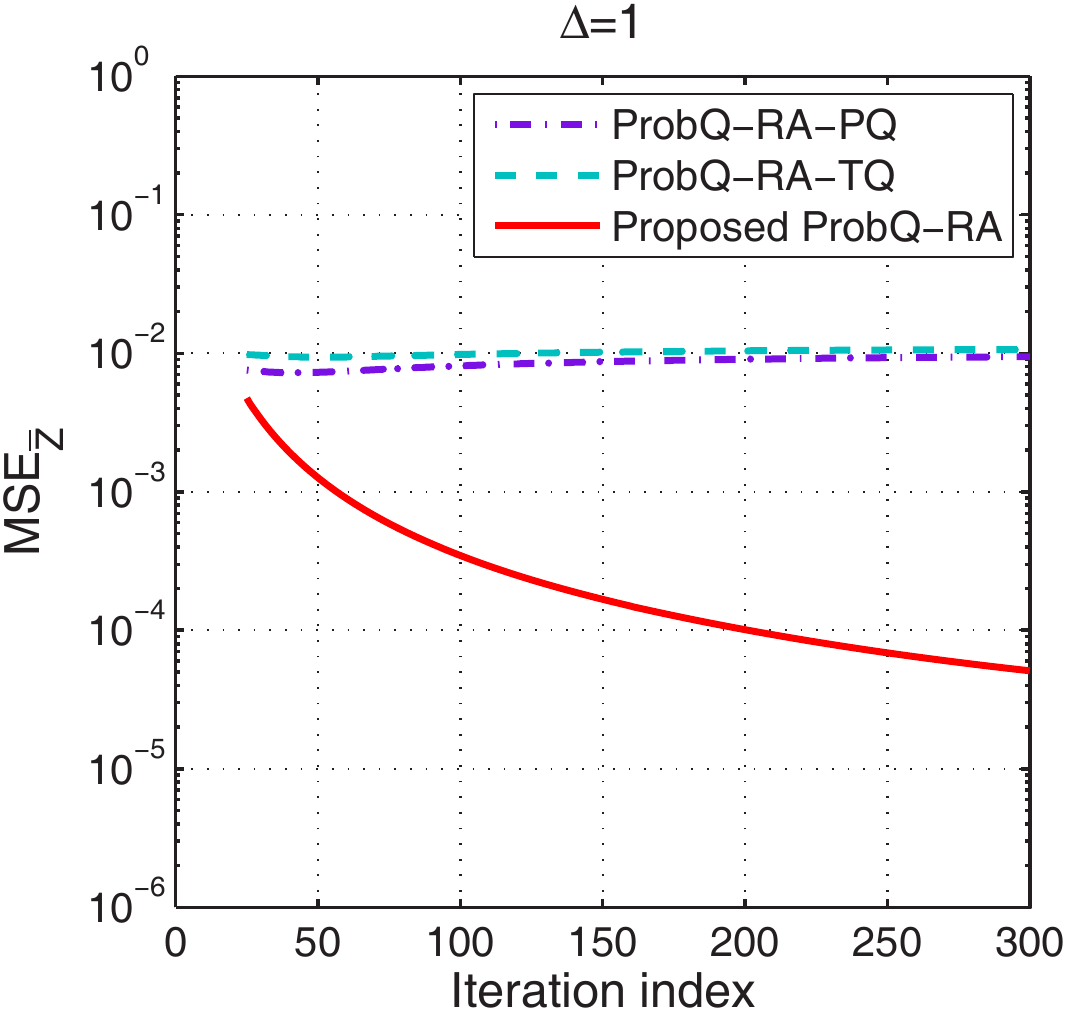}}\\
		\subfloat[]{\includegraphics[width=1.75in]{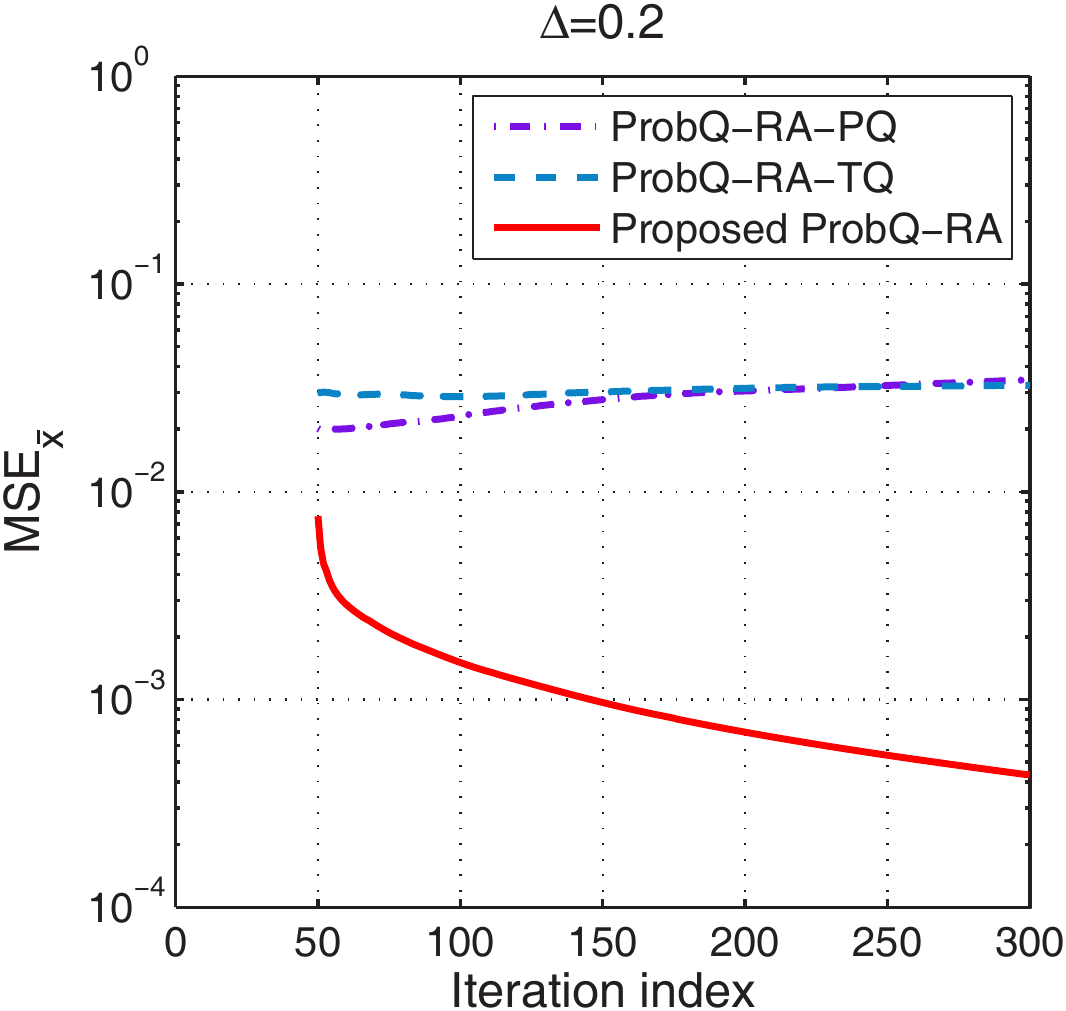}}
		\subfloat[]{\includegraphics[width=1.75in]{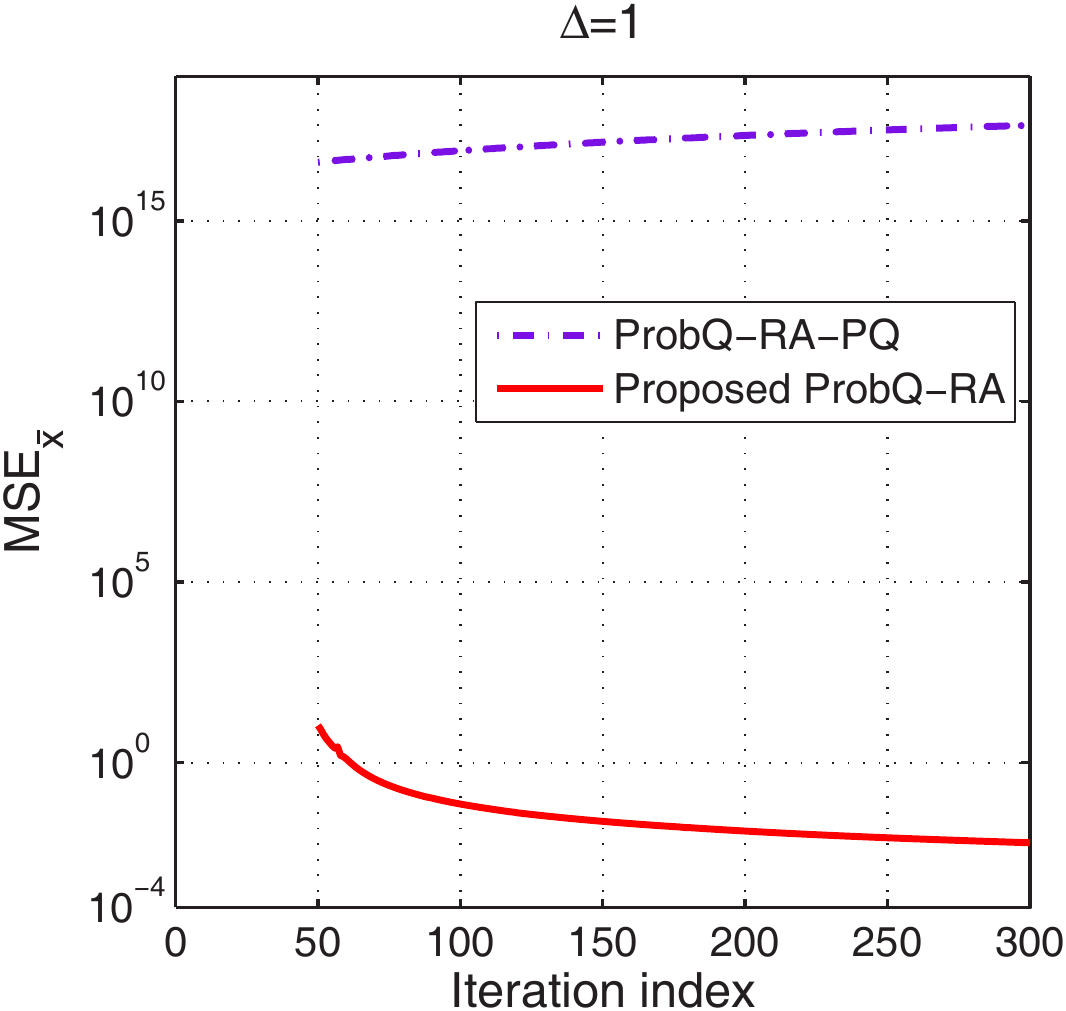}}\\
	\end{center}
	\caption{Comparison of  the mean square errors of ProbQ-RA-PQ, ProbQ-RA-TQ and ProbQ-RA: MSE${}_{\bar{\mathbf{Z}}}$ ((a) and (b)), and MSE${}_{\bar{\mathbf{x}}}$  ((c) and (d)) with respect to $\Delta\in \{0.2,1\}$.}	
	\label{fig:MSEzRA&xRA}
\end{figure} 

\subsection{Comparison with the partially quantized and totally quantized updating rules}
    In Fig.~\ref{fig:MSEzRA&xRA}, we plot the results of the average mean square errors $\text{MSE}_{\bar{\mathbf{Z}}}$ and $\text{MSE}_{\bar{\mathbf{x}}}$ for three updating rules  using  running average, where ProbQ-RA-PQ and ProbQ-RA-TQ denote the averaging based partially quantized (PQ) and totally quantized (TQ) rules \cite{CarFagFraZam10}. From the results, we can see that the averaging based PQ and TQ rules perform well for the left eigenvector estimation 	for both $\Delta=0.2$ and 1.  However, it is observed from Fig.~\ref{fig:MSEzRA&xRA}(c) that the errors are quite large at the second stage even with a rather high quantization resolution $\Delta=0.2$. Moreover, with the quantization resolution decreased from $\Delta=0.2$ to $1$, both PQ and TQ rules do not produce acceptable results (see Fig.~\ref{fig:MSEzRA&xRA}(d))\footnote{As the running average $\bar{z}_{ii}(K)$ of TQ rule doesn't converge and will be zeros many times, the correction term $\epsilon_i(t)$ in \eqref{eq:correction} is meaningless for TQ rule. So we do not provide the data of ProbQ-RA-TQ in Fig.~\ref{fig:MSEzRA&xRA}(d).}: PQ rule diverges and TQ rule doesn't provide any meaningful data for large $\Delta$.  Different from the PQ and TQ rules, the update rule used in \eqref{eq:algorithm} and \eqref{eq:estimateeigenvec} performs quite well for all the cases and the running average can further improve its accuracy. This is consistent with the aforementioned theoretical analysis.

	Finally, we compare UnifQ, ProbQ, ProbQ-RA-PQ, ProbQ-RA-TQ and ProbQ-RA regarding the average mean square error for different quantization resolutions. The results are shown in Fig.~\ref{fig:MSEvsresolution}   (as for ProbQ-RA-TQ, we only plot the results for $\Delta\in\{0.05,0.1,0.2\}$, since no meaningful data can be guaranteed with the same setup as those of the other two updating rules). In order to avoid the transient periods, we take the average of the last 150 iterations of $\text{MSE}_{\bar{\mathbf{x}}}$ in presenting the results. From the figure, we can see that the proposed ProbQ-RA works quite well even when $\Delta=1$. There are significant improvements of the performance at lower quantization resolutions by using ProbQ-RA compared with other algorithms.  The running average technique does improve the performance of PQ and TQ rules for smaller $\Delta$. While for larger $\Delta$, it seems that the running average does not have such effect on PQ and TQ rules.  Although the running average has smoothing effects on random data, the above simulations indicate that only certain kinds of algorithms can benefit from this consequence. 
\begin{figure}[!t]
	\begin{center}
		 \includegraphics[width=2.5in]{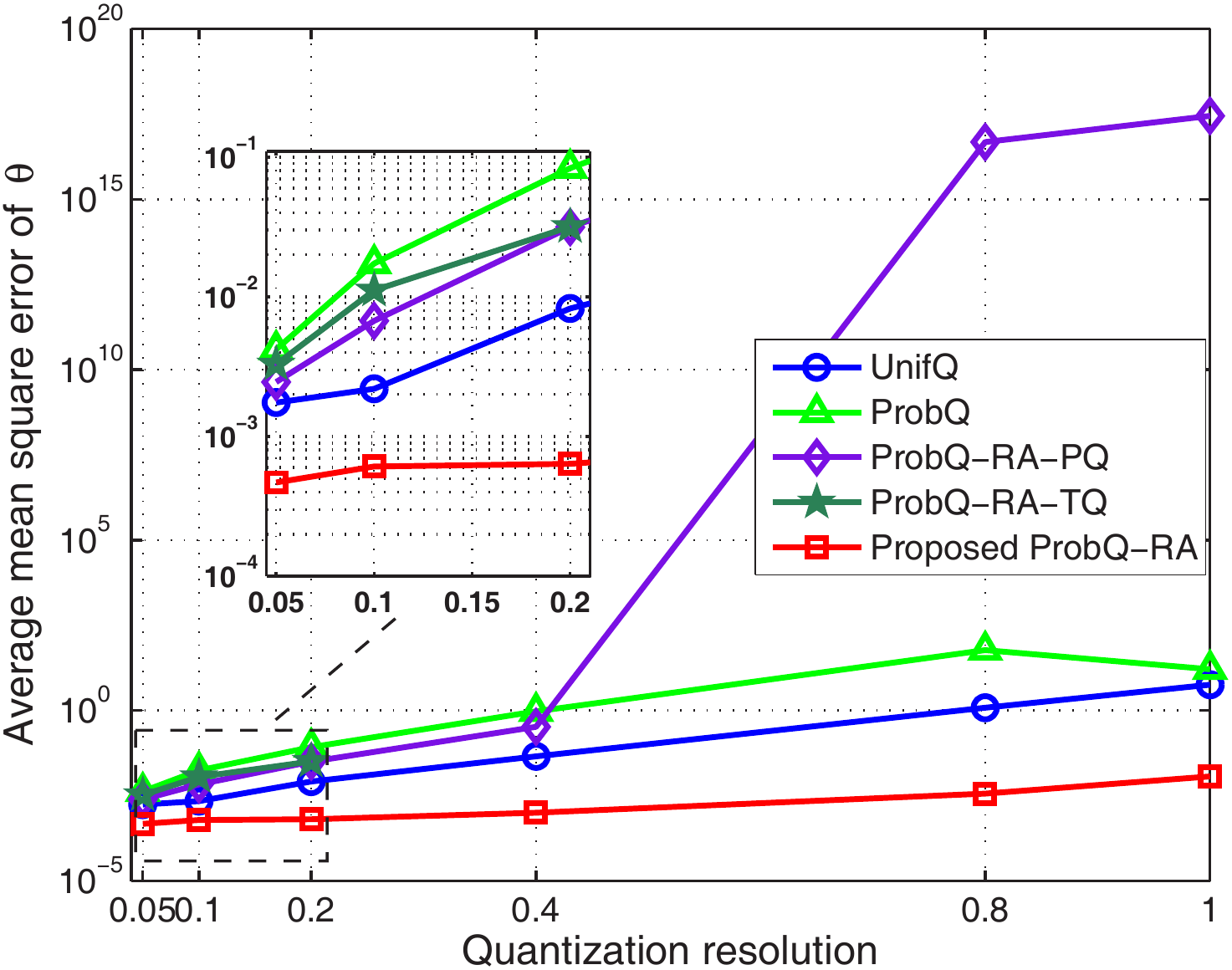}
	\end{center}
	\caption{Average mean square errors of UnifQ, ProbQ, ProbQ-RA, ProbQ-RA-PQ and ProbQ-RA-TQ for different quantization resolutions $\Delta\in \{0.05, 0.1, 0.2, 0.4, 0.8,1\}$.}	
	\label{fig:MSEvsresolution}
\end{figure}

\section{Conclusions and Future Works} 
\label{sec:concluding_remarks}
	We have studied the problem of distributed parameter estimation  over sensor networks in the presence of quantized data and directed communication links. We have proposed a two-stage algorithm such that the centralized sample mean estimate can be achieved in a distributed manner. In the algorithm, the running average technique is utilized  to smear out the randomness caused by the probabilistic quantization scheme. We have shown that the proposed algorithm can achieve the centralized sample mean estimate both in the mean square and almost sure senses.  Finally, we have presented simulation results to illustrate the effectiveness of the proposed algorithm.   Comparisons with other algorithms have also been provided to highlight the improvements of the proposed algorithms.

    Some future directions include the investigation of more efficient algorithms, which are scalable in the network size,  and the effects of other forms of running average on the performance of the algorithm.

\appendices
\section{Proof of Lemma~\ref{lem:matrixproperty}}\label{app:matrixproperty}
	i) Since $\mathcal{G}$ is strongly connected, $\lambda_1(\mathbf{L})=0$ is a simple eigenvalue and $\lambda_i(\mathbf{L})\neq 0$, $\forall i\in\{2,\dots,n\}$. In view of the principle of biorthogonality \cite[p.78]{HorJoh13}, all the eigenvectors $\boldsymbol{\upsilon}_i$ corresponding to $\lambda_i(\mathbf{L})$, $i\in\{2,\dots,n\}$, are orthogonal to $\boldsymbol{\omega}$, that is, $\boldsymbol{\omega}^T\boldsymbol{\upsilon}_i=0$. This implies that $\mathbf{Q}\boldsymbol{\upsilon}_i=(\mathbf{I}-\alpha \mathbf{L})\boldsymbol{\upsilon}_i=(1-\alpha \lambda_i(\mathbf{L})) \boldsymbol{\upsilon}_i$. Hence, $1-\alpha \lambda_i(\mathbf{L})$, $i\in\{2,\dots,n\}$, are eigenvalues of $\mathbf{Q}$. Moreover, it can be verified that $\mathbf{Q}\mathbf{1}=\mathbf{0}$, since $\mathbf{L}\mathbf{1}=\mathbf{0}$ and $\mathbf{1}^T \boldsymbol{\omega}=1$. 

	ii) By i), one has $\rho(\mathbf{Q})=\max_{2\leq i\leq n}|1-\alpha \lambda_i(\mathbf{L})|=\max_{2\leq i\leq n}(\alpha^2 |\lambda_i(\mathbf{L})|^2-2\alpha \text{Re}(\lambda_i(\mathbf{L}))+1)^{1/2}$. Since $\mathcal{G}$ is strongly connected, we know that $\text{Re}(\lambda_i(\mathbf{L}))>0$, $\forall i\in\{2,\dots,n\}$ \cite[Lemma~3.3]{RenBea05}. Hence,  $\rho(\mathbf{Q})<1$ if and only if $0<\alpha<\min_{2\leq i\leq n}\left\{2 \text{Re}(\lambda_i(\mathbf{L}))/|\lambda_i(\mathbf{L})|^2\right\}$.

	iii) It follows from Theorem~1 of \cite{Gau53} that  there is a constant $c_\mathbf{Q}>0$ depending only on $\mathbf{Q}$ so that $\|\mathbf{Q}^k\|_F\leq c_\mathbf{Q} k^{q-1} \sum_{i=1}^n |\lambda_i(\mathbf{Q})|^{k}$.
	It is immediate that $\|\mathbf{Q}^k\|_F\leq nc_\mathbf{Q} k^{q-1}\rho^k(\mathbf{Q})$.

\begin{figure}[!t]
	\begin{center}
		\includegraphics[width=2.5in]{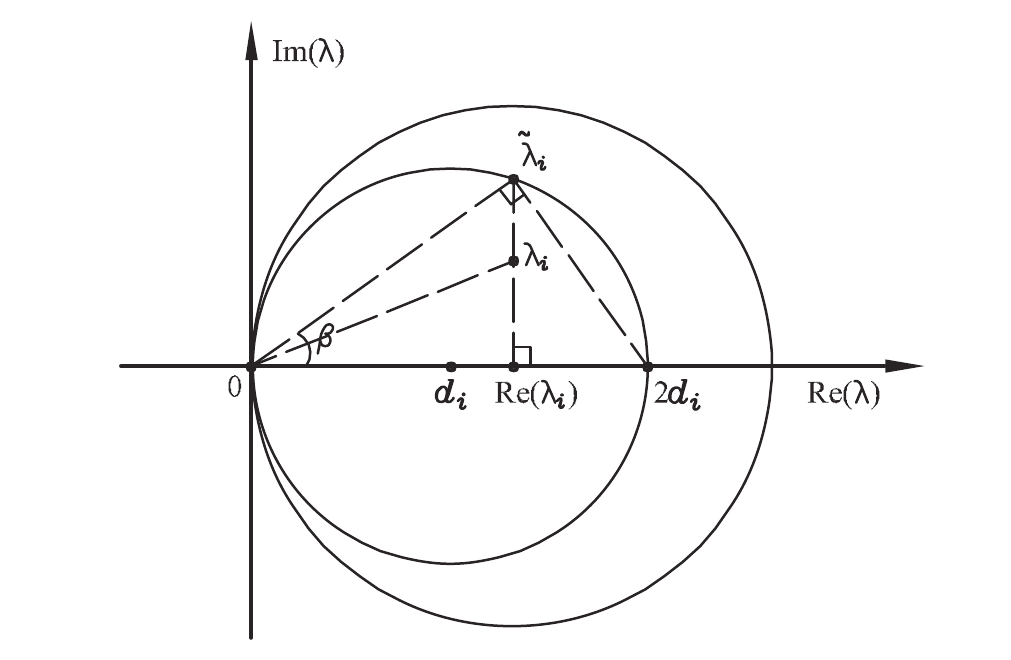}
	\end{center}
	\caption{Illustration of Ger\v{s}gorin discs.}
	\label{fig:gersgorin}
\end{figure}

\section{Proof of Lemma~\ref{lem:parameter}}\label{app:parameter}
	Recall that $\text{Re}(\lambda_i(\mathbf{L}))>0$, $i\in\{2,\dots,n\}$ for strongly connected graphs,  we only need to show that $\max_i d_i\geq \max_{2\leq i\leq n}\left\{|\lambda_i(\mathbf{L})|^2/(2\text{Re}(\lambda_i(\mathbf{L})))\right\}$ in view of Lemma~\ref{lem:matrixproperty}.
	Indeed, by the Ger\v{s}gorin disc theorem \cite[p.388]{HorJoh13}, all the eigenvalues of $\mathbf{L}$ are located in the union of discs $\bigcup_{i=1}^n \bigl\{|\lambda-d_{i}|\leq \sum_{j=1}^n |l_{ij}|\bigr\}$. Consequently, for each $\lambda_i(\mathbf{L})$, $i\in\{2,\dots,n\}$, we can find a $\tilde{\lambda}_i$ located on some circle such that $\text{Re}(\tilde{\lambda}_i)=\text{Re}(\lambda_i(\mathbf{L}))$ and $|\lambda_i(\mathbf{L})|\leq |\tilde{\lambda}_i|$ (see Fig.~\ref{fig:gersgorin}).  This means that 
\begin{equation}\label{eq:gersgorinineq}
	\frac{2\text{Re}(\lambda_i(\mathbf{L}))}{|\lambda_i(\mathbf{L})|^2}\geq \frac{2|\tilde{\lambda}_i|\cos(\beta)}{|\tilde{\lambda}_i|^2}=\frac{2\cos(\beta)}{|\tilde{\lambda}_i|}.
\end{equation}
	On the other hand, we have $\cos(\beta)=|\tilde{\lambda}_i|/ (2d_{i})$, which along with \eqref{eq:gersgorinineq} gives $\min_{2\leq i\leq n}\left\{2\text{Re}(\lambda_i(\mathbf{L}))/|\lambda_i(\mathbf{L})|^2 \right\}\geq (\max_i d_{i})^{-1}$. Hence, by Lemma~\ref{lem:matrixproperty}, we have $\rho(\mathbf{Q})<1$.
    
    It is easy to verify that $\mathbf{P}$ is nonnegative, which implies that $\text{trace}(\mathbf{Q}^k)\geq -1$, $\forall k\in \mathbb{Z}_{\geq 1}$. By Lemma~\ref{lem:matrixproperty}, one then has
\begin{align}
	 \bigl\|\mathbf{I}-\mathbf{Q}^k\bigr\|_F^2&=n-2\ \mbox{tr}(\mathbf{Q}^k) +\|\mathbf{Q}^k\|_F^2\notag\\
	&\leq n+2+n^2 c_\mathbf{Q}^2 k^{2(q-1)}\rho^{2k}(\mathbf{Q}).\label{eq:IQkFbound}
\end{align}
    Consider the case of $q>1$, let $T_*\triangleq (1-q)/ \log \rho(\mathbf{Q})$,  it can be shown that $t^{2(q-1)}\rho^{2t}(\mathbf{Q})$ is monotonically increasing on the interval $(0,T_*]$ and decreasing on $[T_*, \infty)$. Therefore, $k^{2(q-1)}\rho^{2k}(\mathbf{Q})\leq T_*^{2(q-1)}e^{2(1-q)}$, $\forall k\in \mathbb{Z}_{\geq 1}$. Substituting this bound into \eqref{eq:IQkFbound} completes the proof.

\section{Proof of Theorem~\ref{thm:MSconvergence}}\label{app:MSconvergence}
	Using the facts that $\|\mathbf{B}\mathbf{C}\|_F\leq \|\mathbf{C}\|_2\|\mathbf{B}\|_F$ for arbitrary $\mathbf{B}, \mathbf{C}\in \mathbb{R}^{n\times n}$, and $\max_{i} \mathbb{E}\{\|\mathbf{u}_i(t)\|^2\}\leq n\Delta^2/4$ in light of \eqref{eq:varianncequantierror}, one can obtain from \eqref{eq:eZmeansquarebound} that
\begin{align*}
	\mathbb{E}\bigl\{\|\mathbf{e}_{\bar{\mathbf{Z}}}(K)\|_F^2\bigr\}
	&\leq \frac{\|\tilde{\mathbf{Q}}\|_2^2\bigl\|\mathbf{I}-\mathbf{Q}^K\bigr\|_F^2}{K^2}\\
	&\relphantom{\leq}{} + \frac{n\alpha^2\Delta^2\|\tilde{\mathbf{L}}\|_2^2}{4K^2}\sum_{k=1}^{K}\bigl\|\mathbf{I}-\mathbf{Q}^{k}\bigr\|_F^2.
\end{align*}
It thus follows from Lemma~\ref{lem:parameter} that 
\begin{align}
	\mathbb{E}\bigl\{\|\mathbf{e}_{\bar{\mathbf{Z}}}(K)\|_F^2\bigr\}&\leq \frac{n\nu^2}{4K} +\frac{c_{\mathbf{Q},n}^2\|\tilde{\mathbf{Q}}\|_2^2}{K^2}\notag\\
	&\relphantom{\leq}{}+\frac{\nu^2 n^3c_\mathbf{Q}^2}{4(n+2)K^2}\sum_{k=1}^{K} k^{2(q-1)} \rho^{2k}(\mathbf{Q}). \label{eq:MSdebound}
\end{align}

	Since $\rho(\mathbf{Q})<1$ by Lemma~\ref{lem:parameter},  we have $\lim_{k\to \infty} (1+1/k)^{2(q-1)}\rho^2(\mathbf{Q})=\rho^2(\mathbf{Q})<1$, which shows that 
	$\sum_{k=1}^{t} k^{2(q-1)} \rho^{2k}(\mathbf{Q})$ is a convergent series.
	Thus it follows from \eqref{eq:MSdebound} that	for all large $K$,
\[
	\mathbb{E}\left\{\|\mathbf{e}_{\bar{\mathbf{Z}}}(K)\|_F^2\right\}\leq \frac{n\nu^2}{4K}+o\left(K^{-1}\right),
\]
	from which the theorem follows.

\section{Proof of Theorem~\ref{thm:eZas}}\label{app:eZas}
	We will make use of the following result.

	\emph{Lemma D.1:}
	Suppose that $\rho(\mathbf{Q})<1$, then for any $q\in \mathbb{Z}_{\geq 1}$,  the series $\sum_{k=1}^{t} k^{q-1}\rho^{k}(\mathbf{Q})\leq c_{\mathbf{Q}}'$, $\forall t\in \mathbb{Z}_{\geq 1}$, where $c_{\mathbf{Q}}'$ is defined in Theorem~\ref{thm:eZas}.
\begin{proof}
    The case of $q=1$ is straightforward. We only need to consider the case of $q>1$. From the proof of Lemma~\ref{lem:parameter}, we know that  $f(t)\triangleq t^{q-1}\rho^{t}(\mathbf{Q})$ is monotonically increasing on $(0,T_*]$ and decreasing on $[T_*, \infty)$, where $T_*\triangleq (1-q)/ \log \rho(\mathbf{Q})$. By exploiting this monotone property, one can show that
\begin{equation}\label{eq:MSboundintegral}
	\sum_{k=1}^{t} f(k) \leq f(T_*)+\int_{1}^t f(s) \mathrm{d} s.
\end{equation}
	On the other hand, by repeatedly using  integration by parts, we have 
\begin{equation*}
	\int_{1}^t f(s) \mathrm{d} s=\sum_{k=0}^{q-1}(-1)^{q-k+1} \frac{(q-1)!}{k!}\frac{t^{k}\rho^{t}(\mathbf{Q}) - \rho(\mathbf{Q})}{(\log \rho(\mathbf{Q}))^{q-k}}.
\end{equation*}
	Substituting the above relation into \eqref{eq:MSboundintegral} and noting that $t^{k}\rho^{t}(\mathbf{Q})\to 0$ as $t\to \infty$ completes the proof.
\end{proof}

    \emph{Proof of Theorem~\ref{thm:eZas}:}
	By \eqref{eq:boundquantierror}, one has $\|\mathbf{U}(k)\|_F\leq n \Delta$. Let $g_K\triangleq \|\sum_{k=0}^{K-1} \mathbf{U}(k)\|_F$, then 
\begin{equation}\label{eq:LILdecompose}
	\left\|\sum_{k=0}^{K-1}\mathbf{W}_K(k)\tilde{\mathbf{L}}\mathbf{U}(k)\right\|_F \leq \|\tilde{\mathbf{L}}\|_2\left(n \Delta \sum_{k=1}^{K} \|\mathbf{Q}^k\|_F +g_K\right).
\end{equation}
	By Lemmas~\ref{lem:matrixproperty} and D.1, one can obtain   
\begin{equation}
	\sum_{k=1}^{K} \|\mathbf{Q}^k\|_F\leq nc_\mathbf{Q}\sum_{k=1}^{K} k^{q-1}\rho^k(\mathbf{Q})\leq nc_\mathbf{Q} c_{\mathbf{Q}}'. \label{eq:sumQkbound}
\end{equation}
    This implies that the first term of the RHS of \eqref{eq:LILdecompose} is bounded.

	It remains to provide the quantitative bound of $g_K$. By the definition of $\mathbf{U}(t)$, it can be verified that 
\begin{equation}\label{eq:Fro2normequality}
	g_K\leq \sum_{i=1}^n\left\|\sum_{k=0}^{K-1} \mathbf{u}_i(k)\right\|.
\end{equation}

	Now considering each $i\in\{1,\dots,n\}$, we have the next two cases:

	\underline{\emph{Case I.}} $\sup_{K\geq 1}r_K^\mathbf{U}<\infty$. It is obvious that $\mathbb{E}\{\|\mathbf{u}_i(k)\|^2\}=\text{trace}(\text{Cov}(\mathbf{u}_i(k)))$, $\forall k\in \mathbb{Z}_{\geq 0}$. Consequently, we have 
\[
	\sum_{k=0}^{K-1} \mathbb{E}\{\|\mathbf{u}_i(k)\|^2\}=\text{trace}\left(\sum_{k=0}^{K-1}\text{Cov}(\mathbf{u}_i(k))\right)\leq  n\sup_{K\geq 1}r_K^\mathbf{U}.
\]
	Recall that $\{\mathbf{u}_i(k)\}_{k\geq 0}$ is a sequence of independent bounded random vectors. By employing the Kolmogorov three series theorem \cite[p.89]{BulSol97}, we know that
	$\sum_{k=0}^{K-1} \mathbf{u}_i(k)$ converges a.s. as $K\to \infty$. Thus there exists a constant $c_\mathbf{U}>0$ so that 
\begin{equation}\label{eq:boundui}
    \max_{i}\left\|\sum_{k=0}^{K-1} \mathbf{u}_i(k)\right\|\leq c_\mathbf{U} \ \  \text{a.s.},\ \forall K\in \mathbb{Z}_{\geq 1}.
\end{equation}

    Substituting \eqref{eq:sumQkbound}, \eqref{eq:Fro2normequality} and \eqref{eq:boundui} into \eqref{eq:LILdecompose} implies that $\|\sum_{k=0}^{K-1}\mathbf{W}_K(k)\tilde{\mathbf{L}}\mathbf{U}(k)\|_F \leq n(nc_{\mathbf{Q}}c_{\mathbf{Q}}'\Delta +c_\mathbf{U})\|\tilde{\mathbf{L}}\|_2$ a.s.. Moreover, $t^k\rho^{t}(\mathbf{Q})\to 0$ as $t\to \infty$ for all $k\in \mathbb{Z}_{\geq 0}$.  It then follows from \eqref{eq:eZasbound} that for large $K$,
\[
	\|\mathbf{e}_{\bar{\mathbf{Z}}}(K)\|_F\leq \frac{\mu}{K}+o\left(K^{-1}\right)\ \ \mbox{a.s.}
\]

	\underline{\emph{Case II.}} $\lim_{K\to \infty}r_K^\mathbf{U}=\infty$. In this case, there is $i_0\in\{1,\dots,n\}$ satisfying $r_{Ki_0}^\mathbf{U}\triangleq \lambda_{\max}\bigl(\sum_{k=0}^{K-1} \text{Cov}(\mathbf{u}_{i_0}(k))\bigr)\to \infty$ as $K\to \infty$. Hence $\log\log r_{Ki_0}^\mathbf{U}=o(r_{Ki_0}^\mathbf{U})$ as $K\to \infty$.
	On the other hand, one obtains $\|\mathbf{u}_{i_0}(k)\|^2\leq n\Delta^2$, $\forall k\in \mathbb{Z}_{\geq 0}$.  It thus follows from Theorem~1.1 of \cite{Chen93} that 
\begin{equation}\label{eq:asLTL1}
	\limsup_{K\to \infty} \frac{\left\|\sum_{k=0}^{K-1} \mathbf{u}_{i_0}(k)\right\|}{\sqrt{2r_{Ki_0}^\mathbf{U} \log\log r_{Ki_0}^\mathbf{U}}}\leq 1 \ \ \mbox{a.s.}
\end{equation}

	Now, invoking \eqref{eq:Fro2normequality}, \eqref{eq:boundui} and noting that $\|\sum_{k=0}^{K-1}\text{Cov}(\mathbf{u}_i(k))\|_2\leq r_K^\mathbf{U}$, $\forall i$, gives 
\begin{equation}\label{eq:asLIL}
	g_K\leq n c_\mathbf{U}+(n+o(1))\sqrt{2 r_K^\mathbf{U} \log\log r_K^\mathbf{U}}, 
\end{equation}
	for all large $K$.
	Substituting \eqref{eq:LILdecompose}, \eqref{eq:sumQkbound} and \eqref{eq:asLIL} into \eqref{eq:eZasbound}, we finally get for large $K$,
\begin{multline*}
	\|\mathbf{e}_{\bar{\mathbf{Z}}}(K)\|_F\leq \alpha n \|\tilde{\mathbf{L}}\|_2  K^{-1}\sqrt{2 r_K^\mathbf{U} \log\log r_K^\mathbf{U}}\\+o \left( K^{-1}\sqrt{r_K^\mathbf{U} \log\log r_K^\mathbf{U}} \right).
\end{multline*} 

	Combining the above two cases completes the proof.

\section{Proof of Lemma~\ref{lem:correctionconvergence}} \label{app:correctionconvergence}
	By \eqref{eq:correction}, for each $i\in \{1,2,\dots,n\}$,  we have 
\begin{equation}\label{eq:correctionineq}
	|\epsilon_i(t)|=\frac{|y_i|}{n}\frac{|\bar{z}_{ii}(t+1)-\bar{z}_{ii}(t) |}{|\bar{z}_{ii}(t+1)||\bar{z}_{ii}(t)|}, \ \forall t\in \mathbb{Z}_{\geq 1}.
\end{equation}
	It then follows from \eqref{eq:correctionineq} and Theorem~\ref{thm:welldefined} that 
\begin{align*}
	\mathbb{E}\bigl\{\epsilon_i^2(t)\bigr\}&\leq \frac{y_i^2}{n^2 \eta^4\omega_i^4}\mathbb{E}\left\{ |\bar{z}_{ii}(t+1)-\bar{z}_{ii}(t) |^2 \right\}.
\end{align*}
	Note that $\mathbb{E}\left\{ |\bar{z}_{ii}(t+1)-\bar{z}_{ii}(t) |^2 \right\}\leq 2\mathbb{E}\{|\bar{z}_{ii}(t+1)-\omega_i|^2\}+2\mathbb{E}\{|\bar{z}_{ii}(t)-\omega_i|^2\}$,
	we use \eqref{eq:barz-wi_errineq} and Theorem~\ref{thm:MSconvergence} to obtain  
\[
	\mathbb{E}\bigl\{\|\boldsymbol{\epsilon}(t)\|^2\bigr\} \leq \frac{\nu^2 \max_{i} \omega_i^{-4}|y_i|^2}{n\eta^4}\frac{1}{t} +o\left(t^{-1}\right).
\]

	We now turn to the second part. The almost sure convergence follows from \eqref{eq:correctionineq}, Theorems~\ref{thm:eZas} and \ref{thm:welldefined}. To establish the upper bound, by \eqref{eq:correctionineq} and Theorem~\ref{thm:welldefined}, we have
\begin{align*}
	\|\boldsymbol{\epsilon}(t)\|&\leq \Biggl(\sum_{i=1}^{n}\frac{y_i^2}{n^2 \eta^4\omega_i^4} |\bar{z}_{ii}(t+1)-\bar{z}_{ii}(t)|^2\Biggr)^{1/2}\notag\\
	&\leq \frac{\sqrt{2}\max_{i} \omega_i^{-2}|y_i|}{n\eta^2}\bigl( \|\mathbf{e}_{\bar{\mathbf{Z}}} (t+1)\|_F+\|\mathbf{e}_{\bar{\mathbf{Z}}}(t)\|_F \bigr),
\end{align*}
	where the last inequality follows from \eqref{eq:barz-wi_errineq}. 
	Therefore, applying Theorem~\ref{thm:eZas} to the previous relation completes the proof.

\section{Proof of Lemma~\ref{lem:weightsum}} \label{app:MSweightsum}
	By \eqref{eq:compensate}, we can employ the H\"older inequality to obtain
\[
	 e_{\mathbf{x}}^2(t)\leq \frac{1}{n^2}\sum_{i=1}^n \frac{y_i^2}{\bar{z}_{ii}^2(t)}\sum_{i=1}^n \left( \omega_i-\bar{z}_{ii}(t) \right)^2,
\]
	which together with \eqref{eq:barz-wi_errineq} and \eqref{eq:barziibound} implies that  
\[
    |e_{\mathbf{x}}(t)|\leq \frac{y''}{\eta\sqrt{n}} \|\mathbf{e}_{\bar{\mathbf{Z}}}(t)\|_F, \ \ 
	\mathbb{E}\bigl\{ e_{\mathbf{x}}^2(t)\bigr\}\leq \frac{y''^2}{n\eta^2} \mathbb{E}\{\|\mathbf{e}_{\bar{\mathbf{Z}}}(t)\|_F^2\}.
\]
	The lemma thus follows from Theorems~\ref{thm:MSconvergence} and \ref{thm:eZas}.

\section{Proof of Theorem~\ref{lem:MSbarxx_eK}}\label{app:MStilde_eK}
	First, we can obtain from \eqref{eq:tilde_eK} that $\mathbb{E}\bigl\{\|\mathbf{e}_{\bar{\mathbf{x}},\mathbf{x}}(K)\|^2 \bigr\}\leq 3 K^{-2} \left( \mathbb{E}\{\|\mathcal{I}_1\|^2\}+\mathbb{E}\{\|\mathcal{I}_2\|^2\}+\mathbb{E}\{\|\mathcal{I}_3\|^2\}\right)$. 

	Considering $\mathbb{E}\bigl\{\|\mathcal{I}_1\|^2\bigr\}$, by Assumption 2, one has
\begin{equation}\label{eq:MSboundcalI10}
	\begin{split}
	\mathbb{E}\bigl\{\|\mathcal{I}_1\|^2\bigr\}&=\bigl\|\tilde{\mathbf{Q}}(\mathbf{I}-\mathbf{Q}^K)\bigr\|_2^2\|\mathbf{y}\|^2\\
	&\relphantom{=}{} +\alpha^2 \sum_{k=0}^{K-1}\mathbb{E}\left\{\bigl\|\mathbf{W}_K(k)\tilde{\mathbf{L}}\mathbf{v}(k)\bigr\|^2\right\}.
	\end{split}
\end{equation}
	Let $\text{Cov}^{1/2}(\mathbf{v}(t))$ be the square root of $\text{Cov}(\mathbf{v}(t))$, then it follows from \eqref{eq:varianncequantierror} and the relation $\|\text{Cov}^{1/2}(\mathbf{v}(t))\|_2\leq \max_{i} \sqrt{\mathbb{E}\{v_i^2\}}$ that
\begin{align}\label{eq:eq:MSboundcalI11}
	\mathbb{E}\left\{\bigl\|\mathbf{W}_K(k)\tilde{\mathbf{L}}\mathbf{v}(k)\bigr\|^2\right\}&=\bigl\|\mathbf{W}_K(k)\tilde{\mathbf{L}}\text{Cov}^{1/2}(\mathbf{v}(k))\bigr\|_F^2\notag\\
	&\leq  \frac{\Delta^2\|\tilde{\mathbf{L}}\|_2^2}{4} \bigl\|\mathbf{W}_K(k)\bigr\|_F^2.
\end{align}
	Moreover, we derive from Lemma~\ref{lem:parameter} that
\begin{equation}\label{eq:sumtildeW}
	\sum_{k=0}^{K-1}\bigl\|\mathbf{W}_K(k)\bigr\|_F^2
	= \sum_{k=1}^{K}\|\mathbf{I}-\mathbf{Q}^k\|_F^2\leq (n+2) K+O(1),
\end{equation}
   since $\sum_{k=1}^{t} k^{2(q-1)} \rho^{2k}(\mathbf{Q})$ is a convergent series.
   Hence, substituting \eqref{eq:eq:MSboundcalI11} and \eqref{eq:sumtildeW} back into \eqref{eq:MSboundcalI10} yields
\begin{equation}\label{eq:MSboundcalI1}
	\mathbb{E}\bigl\{\|\mathcal{I}_1\|^2\bigr\}\leq 
	\frac{\nu^2 K}{4}+O(1).
\end{equation}

	As for $\mathbb{E}\bigl\{\|\mathcal{I}_{2}\|^2\bigr\}$, by Lemma~\ref{lem:correctionconvergence}, we know that  $\sup_{t\geq 0}\mathbb{E}\bigl\{\|\boldsymbol{\epsilon}(t)\|^2\bigr\}$ is bounded. Hence, there is an integer $k_*>0$ such that 
\[
    \sum_{k=0}^{K-1} \mathbb{E}\{\|\boldsymbol{\epsilon}(k)\|^2\}\leq O(1)+\frac{\nu^2 y'^2}{n \eta^4}\sum_{k=k_*}^{K-1}\frac{1}{k}.
\]
    On the other hand, for any two integers $t_1,t_2\in \mathbb{Z}_{\geq 1}$ with $t_2>t_1$, we have the following relation
\begin{equation}\label{eq:sumharmonicseri}
	\sum_{t=t_1}^{t_2}\frac{1}{t}\leq \int_{t_1}^{t_2} \frac{1}{t} \mathrm{d} t+\frac{1}{t_1}=\log\left(\frac{t_2}{t_1}\right)+\frac{1}{t_1}.
\end{equation}
	It thus follows from the $c_r$ inequality \cite[p.127]{Gut13} and Lemma~\ref{lem:parameter} that for large $K$,
\begin{align}
	\mathbb{E}\bigl\{\|\mathcal{I}_{2}\|^2\bigr\}
	&\leq K c_{\mathbf{Q},n}^2\|\tilde{\mathbf{Q}}\|_2^2\sum_{k=0}^{K-1} \mathbb{E}\bigl\{\|\boldsymbol{\epsilon}(k)\|^2\bigr\}\notag\\
	&\leq  \frac{c_{\mathbf{Q},n}^2\nu^2 y'^2 \|\tilde{\mathbf{Q}}\|_2^2}{n\eta^4}   K \log(K-1)+o\left(K \log K\right).\label{eq:MSboundcalI2}
\end{align}

	Let us turn to $\mathbb{E}\bigl\{\|\mathcal{I}_{3}\|^2\bigr\}$. We can obtain 
\begin{align}
	\mathbb{E}\bigl\{\|\mathcal{I}_{3}\|^2\bigr\}&\overset{(a)}{\leq} \frac{K}{n}\sum_{k=1}^{K}  \mathbb{E}\bigl\{(\mathbf{y}^T\boldsymbol{\varepsilon}_{K}(k))^2\bigr\}\notag\\
	&\overset{(b)}{\leq} \frac{2y''^2 K}{n\eta^4}\sum_{k=1}^K \sum_{i=1}^n \mathbb{E}\{|\bar{z}_{ii}(k)-\bar{z}_{ii}(K) |^2\}\notag\\
	&\overset{(c)}{\leq} \frac{2y''^2 K}{n\eta^4} \Biggl(K \mathbb{E} \left\{ \|\mathbf{e}_{\bar{\mathbf{Z}}}(K)\|_F^2 \right\}\Biggr.\notag\\
	&\relphantom{\leq \frac{2y''^2t}{n\eta^4} \Biggl(}{} \Biggl.+\sum_{k=1}^{K}\mathbb{E} \left\{ \|\mathbf{e}_{\bar{\mathbf{Z}}}(k)\|_F^2\right\}\Biggr)\notag\\
	&\overset{(d)}{\leq} \frac{\nu^2 y''^2}{2\eta^4}K \log K+o\left(K\log K\right),\label{eq:MSboundcalI3}
\end{align}
	where  $(a)$ follows from the $c_r$ inequality \cite[p.127]{Gut13}, $(b)$ follows from  Theorem~\ref{thm:welldefined}, $(c)$ is due to \eqref{eq:barz-wi_errineq}, and $(d)$ is obtained by using Theorem~\ref{thm:MSe_Kbound} and \eqref{eq:sumharmonicseri}.

	Combining the bounds \eqref{eq:MSboundcalI1}, \eqref{eq:MSboundcalI2} and \eqref{eq:MSboundcalI3} all together, after some simplifications, we finally complete the proof.

\section{Proof of Theorem~\ref{thm:ase_Kbound}}\label{app:ase_Kbound}
	We only need to consider $\|\mathbf{e}_{\bar{\mathbf{x}},\mathbf{x}}(K)\|$. First, it is trivial that $\|\mathbf{e}_{\bar{\mathbf{x}},\mathbf{x}}(K)\|\leq K^{-1}\left(\|\mathcal{I}_1\|+\|\mathcal{I}_{2}\|+\|\mathcal{I}_{3}\|\right)$.

	For $\|\mathcal{I}_1\|$, by Lemma~\ref{lem:parameter}, we know that $\|\mathbf{I}-\mathbf{Q}^k\|_2\leq \|\mathbf{I}-\mathbf{Q}^k\|_F\leq c_{\mathbf{Q},n}$, $\forall k\in \mathbb{Z}_{\geq 1}$. Let $h_K\triangleq \|\sum_{k=0}^{K-1} \mathbf{v}(k)\|$, one can obtain from \eqref{eq:tilde_eK} that
\begin{align}
	\|\mathcal{I}_1\|&\leq c_{\mathbf{Q},n}  \|\tilde{\mathbf{Q}}\mathbf{y}\|+\alpha\|\tilde{\mathbf{L}}\|_2 \left( \Delta \sum_{k=1}^{K} \|\mathbf{Q}^k\|_2 +h_K\right)\notag\\
	&\leq  \alpha \|\tilde{\mathbf{L}}\|_2 h_K+O(1),\label{eq:ascalI1}
\end{align}
	where we use  \eqref{eq:boundquantierror} and \eqref{eq:sumQkbound}.

	To establish the rate of convergence of $h_K$, we use  similar arguments as that of the proof of Theorem~\ref{thm:eZas}. We consider two cases, separately.  

	\underline{\emph{Case Ia.}} $\sup_{K\geq 1}r_K^{\mathbf{v}}<\infty$. Note that $\mathbf{v}(t)$ is uniformly bounded in light of \eqref{eq:boundquantierror}. Under Assumption 2, the Kolmogorov three series theorem for random vectors \cite[p.89]{BulSol97} applies, and we know that $\sum_{k=0}^{K-1} \mathbf{v}(k)$ converges almost surely as $K$ tends to $\infty$. In particular, there exists a constant $c_\mathbf{v}>0$ so that for all integers $K\in \mathbb{Z}_{\geq 1}$, we have $h_K\leq c_{\mathbf{v}}$ a.s.

	\underline{\emph{Case Ib.}} $\lim_{K\to \infty}r_K^{\mathbf{v}}=\infty$. In this case, similar to \eqref{eq:asLTL1}, one has
\[
	\limsup_{K\to \infty} \frac{h_K}{\sqrt{2r_{K}^\mathbf{v} \log\log r_{K}^\mathbf{v}}}\leq 1 \ \ \text{a.s.}
\]
	Substituting the above two cases into \eqref{eq:ascalI1} implies that $\|\mathcal{I}_1\|$ is approximately bounded by 
\begin{equation}\label{eq:ascalI1bound}
	\begin{split}
	\|\mathcal{I}_1\|
	&\leq 
	\begin{cases} \alpha  c_{\mathbf{v}}\|\tilde{\mathbf{L}}\|_2+O(1), & \sup r_K^{\mathbf{v}}<\infty,\\
	\alpha\|\tilde{\mathbf{L}}\|_2 \sqrt{2r_{K}^{\mathbf{v}} \log\log r_{K}^{\mathbf{v}}} , & \text{otherwise}.
	\end{cases}
	\end{split}
\end{equation}

	Now turning to $\|\mathcal{I}_{2}\|$, similar to \eqref{eq:MSboundcalI2}, we can get  
\begin{equation}\label{eq:ascalI2bound}
	\|\mathcal{I}_{2}\|\leq  c_{\mathbf{Q},n}\|\tilde{\mathbf{Q}}\|_2\sum_{k=0}^{K-1} \|\boldsymbol{\epsilon}(k)\|+ O(1).
\end{equation}
   As for $\|\mathcal{I}_{3}\|$, we can show that  $\|\mathbf{1}\mathbf{y}^T\boldsymbol{\varepsilon}_K(k)\|\leq \sqrt{2n} \eta^{-2}y''\bigl( \|\mathbf{e}_{\bar{\mathbf{Z}}}(k)\|_F+\|\mathbf{e}_{\bar{\mathbf{Z}}}(K)\|_F \bigr)$,
	which together with Theorem~\ref{thm:eZas} implies 
\begin{equation}\label{eq:ascalI3bound}
  \|\mathcal{I}_{3}\|\leq \frac{\sqrt{2} y''}{\eta^2 \sqrt{n}}\Biggl( K\|\mathbf{e}_{\bar{\mathbf{Z}}}(K)\|_F +\sum_{k=1}^{K}\|\mathbf{e}_{\bar{\mathbf{Z}}}(k)\|_F\Biggr).
\end{equation}

	Let $\varpi\triangleq 2 c_{\mathbf{Q},n}y'\|\tilde{\mathbf{Q}}\|_2+\sqrt{n}y''$, then applying Theorem~\ref{thm:eZas} and Lemma~\ref{lem:correctionconvergence} to \eqref{eq:ascalI2bound} and \eqref{eq:ascalI3bound} leads to the next two cases:

	\underline{\emph{Case IIa.}} $\sup_{K\geq 1} r_K^{\mathbf{U}}<\infty$. In this case, we can obtain
\[
	\|\mathcal{I}_{2}+\mathcal{I}_{3}\|\leq 
	\frac{\sqrt{2}\mu \varpi }{n\eta^2 } \log K+o(\log K).
\]

    \underline{\emph{Case IIb.}} $\lim_{K\to \infty}r_{K}^\mathbf{U}=\infty$.
	 Noting that $r_{K}^\mathbf{U} \log\log r_{K}^\mathbf{U}$ is  monotonically increasing of $K$, one has
\begin{multline*}
	\|\mathcal{I}_{2}+\mathcal{I}_{3}\|\leq \frac{2\alpha \varpi\|\tilde{\mathbf{L}}\|_2}{\eta^2} \sqrt{r_{K}^\mathbf{U} \log\log r_{K}^\mathbf{U}}\log K\\ +o\left(\sqrt{r_{K}^\mathbf{U} \log\log r_{K}^\mathbf{U}} \log K \right).
\end{multline*}

	Based on the above discussion, we have the next four cases about the approximate upper bounds of $\|\mathbf{e}_{\bar{\mathbf{x}},\mathbf{x}}(K)\|$, i.e., Case Ia\&IIa, Case Ia\&IIb, Case Ib\&IIa, Case Ib\&IIb.
	Note that $(r_{K}^\mathbf{U} \log\log r_{K}^\mathbf{U})^{1/2}$, $(r_K^{\mathbf{v}} \log\log r_K^{\mathbf{v}})^{1/2}$ are both increasing functions of $K$. The above four cases together with Lemma~\ref{lem:weightsum} complete the proof. 

\section*{Acknowledgment}
The authors would like to thank the anonymous reviewers for their
valuable comments and suggestions to improve the quality and
clarity of this manuscript.

\ifCLASSOPTIONcaptionsoff
\newpage
\fi


\begin{IEEEbiography}
%[{\includegraphics[width=1in,height=1.25in,clip,keepaspectratio]{syzhu}}]
{Shanying Zhu}
received the B.S. degree in information
and computing science from North China University of Water Resources and Electric Power, Zhengzhou, China, in 2006, the M.S. degree in
applied mathematics from Huazhong University of Science and Technology, Wuhan, China, in 2008, and the Ph.D. degree in control theory and control
engineering from Shanghai Jiao Tong University, Shanghai, China, in 2013. Since July 2013, he has been a research fellow with the School of Electrical
and Electronic Engineering, Nanyang Technological University, Singapore. His research interests focus on multi-agent systems and wireless sensor networks, particularly in coordination control of mobile robots and distributed detection and estimation in sensor networks and their applications in industrial networks.
\end{IEEEbiography}

\begin{IEEEbiography}
%[{\includegraphics[width=1in,height=1.25in,clip,keepaspectratio]{ycsoh}}]
{Yeng Chai Soh}
(M'87-SM'06) received the B.Eng. (Hons. I) degree in electrical and electronic engineering from the University of Canterbury, New Zealand, and the Ph.D. degree in electrical engineering from the University of Newcastle, Australia. He joined the Nanyang Technological University, Singapore, after his PhD study and is currently a professor in the School of Electrical and Electronic Engineering. Dr Soh has served as the Head of the Control and Instrumentation Division, the Associate Dean (Research and Graduate Studies) and the Associate Dean (Research) at the College of Engineering. Dr Soh has served as panel members of several national grants and scholarships evaluation committees.

Dr Soh's current research interests are in robust control and applications, robust estimation and filtering, optical signal processing, and energy efficient systems. He has published more than 230 refereed journal papers in these areas. His research achievements in optical signal processing have won him several national and international awards. 
\end{IEEEbiography}

\begin{IEEEbiography}
%[{\includegraphics[width=1in,height=1.25in,clip,keepaspectratio]{lhxie}}]
{Lihua Xie}
(S'91-M'92-SM'97-F'07) received the B.E. and M.E. degrees in electrical engineering from Nanjing University of Science and Technology in 1983 and 1986, respectively, and the Ph.D. degree in electrical engineering from the University of Newcastle, Australia, in 1992. Since 1992, he has been with the School of Electrical and Electronic Engineering, Nanyang Technological University, Singapore, where he is currently a professor and served as the Head of Division of Control and Instrumentation from July 2011 to June 2014. He held teaching appointments in the Department of Automatic Control, Nanjing University of Science and Technology from 1986 to 1989 and Changjiang Visiting Professorship with South China University of  Technology from 2006 to 2011. 

Dr Xie's research interests include robust control and estimation, networked control systems, multi-agent networks, and unmanned systems. In these areas, he has published over 260 journal papers and co-authored two patents and six books. He has served as an editor of IET Book Series in Control and an Associate Editor of a number of journals including IEEE Transactions on Automatic Control, Automatica, IEEE Transactions on Control Systems Technology, and IEEE Transactions on Circuits and Systems-II. Dr Xie is a Fellow of IEEE and Fellow of IFAC.
\end{IEEEbiography}

\end{document}